\documentclass[a4paper,USenglish,cleveref, autoref, thm-restate,authorcolumns, pagebackref]{lipics-v2019}
  \usepackage{amsthm,amsmath,amssymb,graphicx,dsfont}
  \usepackage[capitalise]{cleveref}
  \usepackage{enumitem}
  \usepackage{xspace}

  \crefname{enumi}{Property}{Properties}

  \usepackage{mathtools}
  \usepackage{subcaption}
\usepackage[numbers]{natbib}

\usepackage{dsfont}
\usepackage{tikz}
\usetikzlibrary{graphs,calc,backgrounds,fit,positioning}
\tikzstyle{vertex}=[draw, circle, fill, inner sep = 2pt]

\crefname{subsection}{Subsection}{Subsections}
\Crefname{subsection}{Subsection}{Subsections}

\newcommand{\Nat}{\mathbb{N}}

\newcommand{\Wilog}{Without loss of generality}
\newcommand{\wilog}{without loss of generality}

\newcommand{\Ff}{\mathcal{F}}
\newcommand{\Gg}{\mathcal{G}}
\newcommand{\Ss}{\mathcal{S}}
\newcommand{\Xx}{\mathcal{X}}
\newcommand{\Pp}{\mathcal{P}}
\newcommand{\wx}[1]{\mathrm{W[#1]}\xspace}
\newcommand{\wone}{\ensuremath{\wx{1}}}
\newcommand{\wtwo}{\ensuremath{\wx{2}}}
\newcommand{\fpt}{\ensuremath{\mathrm{FPT}}}
\newcommand{\p}{\ensuremath{\mathrm{P}}}
\newcommand{\np}{\ensuremath{\mathrm{NP}}}

\newtheorem*{problem}{Problem}

\DeclareMathOperator{\poly}{poly}

\DeclarePairedDelimiter{\abs}{\lvert}{\rvert}

\title{Multistage Graph Problems on a Global Budget}
\authorrunning{Multistage Graph Problems on a Global Budget }

\newcommand{\tuaddress}{TU Berlin, Faculty IV, Algorithmics and Computational Complexity, Germany}
\author{Klaus Heeger}{\tuaddress}{heeger@tu-berlin.de}{https://orcid.org/0000-0001-8779-0890}{Funded by  DFG Research Training Group 2434, ``Facets of Complexity''.}
\author{Anne-Sophie Himmel}{\tuaddress}{anne-sophie.himmel@tu-berlin.de}{https://orcid.org/0000-0001-7905-7904}{Funded by DFG, project FPTinP (NI 369/16).}
\author{Frank Kammer}{THM, University of Applied Sciences Mittelhessen,
Giessen, Germany}{frank.kammer@mni.thm.de}{https://orcid.org/0000-0002-2662-3471}{}
\author{Rolf Niedermeier}{\tuaddress}{rolf.niedermeier@tu-berlin.de}{https://orcid.org/0000-0003-1703-1236}{}
\author{Malte Renken}{\tuaddress}{m.renken@tu-berlin.de}{https://orcid.org/0000-0002-1450-1901}{Funded by the DFG, project MATE (NI 369/17).}
\author{Andrej Sajenko}{THM, University of Applied Sciences Mittelhessen, 
Giessen, Germany}{andrej.sajenko@mni.thm.de}{https://orcid.org/0000-0001-5946-8087}{Funded by the DFG -- 379157101.}

\authorrunning{K. Heeger, A.-S. Himmel, F. Kammer, R. Niedermeier, M. Renken, A. Sajenko}%TODO mandatory. First: Use abbreviated first/middle names. Second (only in severe cases): Use first author plus 'et al.'

\Copyright{Klaus Heeger, Anne-Sophie Himmel, Frank Kammer, Rolf Niedermeier, Malte Renken, Andrej Sajenko}%TODO mandatory, please use full first names. LIPIcs license is "CC-BY";  http://creativecommons.org/licenses/by/3.0/

%TODO mandatory; please add comma-separated list of keywords

\category{}%optional, e.g. invited paper

\funding{The work started while Frank Kammer and Andrej Sajenko were visiting TU Berlin. 
}%optional, to capture a funding statement, which applies to all authors. Please enter author specific funding statements as fifth argument of the \author macro.

\nolinenumbers
\hideLIPIcs

\EventEditors{}
\EventNoEds{0}
\EventLongTitle{}
\EventShortTitle{}
\EventAcronym{}
\EventYear{}
\EventDate{}
\EventLocation{}
\EventLogo{}
\SeriesVolume{}
\ArticleNo{}

\begin{document}

\maketitle
\begin{abstract}
Time-evolving or temporal graphs gain more and more popularity when 
studying the behavior of complex networks. In this context, the multistage 
view on computational problems is among the most natural frameworks. Roughly speaking, herein one studies the different 
(time) layers of a temporal graph (effectively meaning that the edge set
may change over time, but the vertex set remains unchanged), and one searches for a solution of 
a given graph problem for each layer. The twist in the multistage setting is that
the solutions found must not differ too much between subsequent layers.
We relax on this already established notion by introducing a global instead of the local budget view studied so far.
More specifically, we allow for few disruptive changes between subsequent layers
but request that overall, that is, summing over all layers, the degree of
change is moderate.
Studying several classical graph problems (both NP-hard and polynomial-time solvable 
ones) from a parameterized complexity angle, we encounter both fixed-parameter tractability and
parameterized hardness results. Somewhat surprisingly, we find that sometimes the 
global multistage versions of NP-hard problems such as \textsc{Vertex Cover} 
turn out to be computationally more tractable than the ones of polynomial-time 
solvable problems such as \textsc{Matching}.
\end{abstract}

\section{Introduction}
Recognizing the need to address the continuous evolution of networks 
and the steady demand for maintenance
due to instances changing over time, in~2014 Eisenstat et al.~\cite{EisenstatMS14} and Gupta 
et al.~\cite{GuptaTW14} introduced 
what is now known as the ``multistage view'' on combinatorial optimization problems.
Focusing on graphs, roughly speaking the idea is to consider a series of
graphs over a fixed vertex set and changing edge set (this is
known as the standard model of 
temporal graphs), and the goal is to find for each 
graph of the series (called a layer or snapshot of the temporal graph)
a solution of the studied computational problem  where solutions to subsequent layers
are not ``too different'' from each other.\footnote{Without going into any details, we remark that there are similarities between ``parameterized 
multistaging'' and parameterized studies of dynamic problems~\cite{AEFRS15,HN13,KST18,LMNN18} and reoptimization~\cite{BBRR18}.} For instance, consider the famous 
\textsc{Vertex Cover} problem. Here, the goal is to find for each layer of the temporal
graph a (small) set of vertices covering all edges, guaranteeing that these vertex sets between two
subsequent layers differ not too much (the degree of change is upper-bounded by a given parameter)~\cite{Fluschnik2019}.
Thus, this can be interpreted as a conservative (no dramatic changes allowed for the solution sets)
view on solving problem instances that evolve over time.
Clearly, the static case (no ``evolution'' takes place, that is, there is only one layer)
is a special case, leading to many computational hardness results in this setting
based already on the hardness of the static version.

Since the pioneering works of Eisenstat et al.~\cite{EisenstatMS14} and Gupta
et al.~\cite{GuptaTW14}, who mainly focused on polynomial-time approximation algorithms,
there has been quite some further development
in studying multistage versions of computational problems.
For instance, there have been recent studies on \textsc{Multistage 
Matching}~\cite{BampisELP18}, \textsc{Multistage Vertex Cover}~\cite{Fluschnik2019}, \textsc{Multistage $s$-$t$ Path}~\cite{FNSZ20}, \textsc{Multistage Knapsack}~\cite{BET19MFCS}, 
\textsc{Online Multistage Subset Maximization}~\cite{BEST19}, and 
\textsc{Multistage Committee Elections}~\cite{BFK20}.
In particular, due to natural parameterizations in this problem setting such as 
``number of layers'' or ``maximum degree of change'' between the solutions for
subsequent instances, also parameterized complexity studies have recently 
been started~\cite{Fluschnik2019}.

We modify the meanwhile standard multistage model by moving from a local to a global perspective 
on the number of allowed changes between solution sets for subsequent instances. 
Whereas in the original model there is a parameter upper-bounding the 
maximum degree of change between every pair of subsequent layers of a temporal graph,
we now introduce a more global view by only upper-bounding the sum of changes.
Intuitively, one may say while we still keep the evolutionary view on dynamically changing
instances and corresponding solutions, our new model allows for occasional 
disruptive changes between subsequent layers while on average the degree of change 
shall be limited.\footnote{We are only aware of a Bachelor Thesis~\cite{Rohm2018} supervised by the TU~Berlin 
group where such a global multistage view has been adopted for the specific case of
\textsc{Vertex Cover}.}

After providing the formal definitions, 
we will review our results in Table~\ref{tab-results}.
First, however, we discuss the main findings of our work.
We provide results both for classical NP-hard graph problems (including a
global multistage version of \textsc{Vertex Cover}) and 
classical problems solvable in polynomial time (including a
global multistage version of \textsc{Matching}).
We consider three central parameters: $k$, upper-bounding the size of the solution for a layer;
$\ell$, the global budget upper-bounding the total number of changes between all layer solutions;
$\tau$, the number of layers (equivalently, the lifetime of the temporal graph).
The three key messages of our work are as follows:
\begin{enumerate}
\item We encounter (parameterized) computational hardness results for the single
parameters~$k$ and $\ell$, and even some combinations of $k, \ell$ and $\tau$;
the main technical contribution here is to show W-hardness results for the parameter~$k$.

\item There is a systematic
algorithmic approach that may lead to fixed-parameter tractability 
results with respect to the combined parameter~$k+\ell$.
% for monotone problems
% (i.e., problems for which each superset of a solution is also a solution).
Indeed, along these lines we also obtain 
polynomial problem kernels with respect to the combined parameter~$k+\tau$.
We exemplify this approach by providing corresponding results for 
global multistage versions of the NP-hard problems \textsc{Vertex Cover},
\textsc{Cluster Editing}, \textsc{Cluster Edge Deletion},
and \textsc{Path Contraction}. 
Notably, we spot a close link to the concept of full kernels from
parameterized enumeration~\cite{DAMASCHKE2006337}.
Moreover, for each of our positive algorithmic results we additionally specify 
the (bit) space complexity of the corresponding algorithms.

\item Global multistage versions of polynomial-time solvable problems 
such as \textsc{Matching}, $s$-$t$-\textsc{Path}, and \textsc{$s$-$t$-Cut} turn out to 
be computationally harder than global multistage versions of the above NP-hard problems.
More specifically, we spot W-hardness results for the combined parameter $k+\ell$, and in the case of \textsc{$s$-$t$-Path}, even for the combined parameter $k + \tau$. This contrast with corresponding, very recent studies of 
$s$-$t$-\textsc{Path}
in the standard multistage model~\cite{FNSZ20}. 
\end{enumerate}
In summary, our first systematic study of global multistage problems 
leads to a rich and promising new scenario in the fast growing field of
studying temporal graph problems~\cite{HS13,HS19}.

\begin{table}[t]
\centering
\caption{
	Overview on our (time) complexity results, where $n$ is the number of vertices of the temporal graph, $k$ is the solution size, $\ell$~is the global budget, and $\tau$~is the lifetime of the temporal graph.
	Note that the first five problems are NP-hard in the static case, while the last three are polynomial-time solvable.
	``?'' denotes open cases. $^\dagger$\, indicates that hardness prevails even for $\ell = 0$, and therefore also for the classical multistage version.
	$^\ddagger$\, indicates that hardness prevails even for planar underlying graphs. Herein, $\wone$-h.\ and $\wtwo$-h.\ refer to parameterized hardness, poly.~kernel refers to the existence of a polynomial-size problem kernel, and para-NP-h.\ refers to NP-hardness even for constant parameter values.
}\label{tab-results}
\begin{tabular}{l|cccc}
  Problem & $k$ & $\ell$ & $k + \ell$ & $k + \tau$ \\
  \hline
  \textsc{Vertex Cover} 			& $\wone$-h. & para-NP-h. & $k^{O(k+\ell)}\tau\rm{poly}(n)$ & poly.\ kernel\\
  \textsc{Path Contraction} 		& $\wone$-h. & para-NP-h. & $k^{O(k+\ell)}\tau\rm{poly}(n)$  &  poly.\ kernel  \\
  \textsc{Cluster Editing} 	& $\wone$-h. & para-NP-h. & $1.82^{k(\ell + 1)} \tau n^3$ & ?\\
  \textsc{Cluster Edge Deletion} 	& $\wone$-h. & para-NP-h. & $1.82^{k(\ell + 1)} \tau n^3$ & ?\\
  \textsc{Planar Dominating Set} 	& $\wtwo$-h. & para-NP-h. & $\wtwo$-$\rm{h.}^\dagger$ & ?\\
  \textsc{Edge Dominating Set} & $\wtwo$-h.$^\ddagger$ & para-NP-h.$^\ddagger$ & $\wtwo$-$\rm{h.}^{\dagger\ddagger}$ & ?\\
  \hline
   \textsc{$s$-$t$-Path} 			& $\wone$-h. & para-NP-h. & $\wone$-$\rm{h.}^\dagger$ & \wone-h.$^\dagger$\\
   \textsc{$s$-$t$-Cut} 			& $\wtwo$-h. & para-NP-h. & $\wtwo$-$\rm{h.}^\dagger$ & ? \\
   \textsc{Matching} 				& $\wone$-h. & para-NP-h. & $\wone$-$\rm{h.}^\dagger$ & ?\\
  \hline
\end{tabular}
\end{table}

For classical multistage problems, research mainly focused on approximation algorithms \cite{BampisELP18,BET19MFCS,GuptaTW14}.
While for \textsc{Multistage Perfect Matching} no polynomial-time $O(n^{1-\epsilon})$-approximation is possible unless $\p = \np$ \cite{BampisELP18,GuptaTW14}, \textsc{Multistage Vertex Cover} is 2-approximable in analogy to its static version \cite{BampisEK19}.
On the parameterized complexity side, \textsc{Multistage Vertex Cover} was shown to be fixed-parameter tractable parameterized by $k + \tau$, while being \wone-hard parameterized by~$k$ even if the symmetric difference of the vertex covers in two successive layers may be at most two \cite{Fluschnik2019}.
\textsc{Global Multistage Vertex Cover} behaves similarly, being fixed-parameter tractable parameterized by $k + \tau$ and \wone-hard parameterized by $k$.
In addition, we show that \textsc{Global Multistage Vertex Cover} is fixed-parameter tractable parameterized by~$k + \ell$.
Bampis et al.~\cite{BampisEK19} showed that a multistage version of \textsc{Min-Cut} is still polynomial-time solvable;
note that their definition of \textsc{Multistage Min-Cut} differs from ours as we consider a cut to be a set of edges, while they consider a cut to be a set of vertices (these two notions are equivalent except for the number of changes between two solutions).

Our work is structured as follows. 
In Section~\ref{sect:prelim}, we provide basic notation and definitions.
In Section~\ref{sect:fpt}, we present our general approach for gaining fixed-parameter tractability results for global multistage versions of NP-hard problems.
Moreover, we present corresponding parameterized hardness results in \Cref{sect:w-hardness},
showing that our fixed-parameter tractability results with combined parameter 
most likely cannot be improved to the parameterization 
by the single parameter solution size.
%These complement our framework by proving tightness in the sense that
%an adaptation to the single parameter $k$
%are impossible under the assumption~$\fpt \not =\wone$.
In Section~\ref{sect:hardness_of_polytime}, we then study global multistage versions
of polynomial-time solvable problems and encounter simple but surprising 
parameterized hardness results, contrasting our more positive results
for NP-hard problems in Section~\ref{sect:fpt}. We conclude in Section~\ref{sect:conclusion}.

\section{Preliminaries}\label{sect:prelim}

For an undirected graph $G=(V,E)$, let $V(G) :=V$ and $E(G) :=E$. For~$v\in V$,
let $E(G(v)) := \{ \{v, u\} \in E(G)\}$ be the edges incident to $v$ in $G$ and let 
$N(v)$ be the (open) neighborhood of~$v$, i.e., $N(v) := \{u : \{v, u\} \in E(G)\}$.
A graph is \emph{planar} if it can be embedded in the plane, that is,
the graph can be drawn in the plane without any crossing edges.   
For any natural number~$x$, let $[x]:= \{1, \ldots x\}$.

\begin{definition}[Temporal Graph]

	We represent a \emph{temporal graph} $\mathcal{G}$ using an ordered sequence of static
	 graphs (called \emph{layers}): $\mathcal{G} = \langle G_1, G_2, \ldots, G_\tau \rangle$.
	 The subscripts $i \in [\tau]$ indexing the graphs in the sequence are the discrete time steps $1$ to $\tau$, where $\tau$ is known as the \emph{lifetime} of $\mathcal G$.
	 The \emph{underlying graph} $G$ of $\mathcal G$ is defined as $V(G) := V(G_1) = \ldots = V(G_\tau)$ and $E(G) := E(G_1) \cup \ldots \cup E(G_\tau)$.
\end{definition}

We consider graph problems $\Xx$ in which we are given a graph $G$ and a nonnegative integer~$k$,
and we are searching for a set of at most $k$ elements that satisfies some property $\Pp_\Xx$ in $G$.
As an example consider the \textsc{Vertex Cover} problem, in which we search for a set of at most $k$ vertices 
with the property $\Pp_{\textsc{Vertex Cover}}$ that these vertices cover all edges of $G$.
Such a problem $\Xx$ is called {\em monotone} if every superset
of a solution still satisfies $\Pp_\Xx$
(\textsc{Vertex Cover} being an example of a monotone problem).

We now generalize such problems from graphs to temporal graphs as follows.

\newcommand{\gmp}{\textsc{Global Multistage Problem}}
\begin{definition}[\gmp{} for graph property $\mathcal{P_X}$]
  Given a triple $({\mathcal G},k,\ell)$ where 
         $\mathcal G = \langle G_1, \ldots, G_{\tau} \rangle$ is a temporal graph,
        the goal is to
        find sets $S_1,\ldots,S_{\tau}$, each of size at most $k$, such that $S_i$
        satisfies $\mathcal{P_X}$ in $G_i$ ($i \in [\tau]$) and the total number of
        insertions from $S_i$ to $S_{i+1}$ over all $i \in [\tau-1]$
        is upper-bounded by $\ell$, that is, $\sum_{i\in [\tau -1 ]} |S_{i+1}
        \setminus S_{i}| \leq~\ell$.
\label{def:globalmultistage}
\end{definition}
   		Note that $k$ upper-bounds the solution size in each step whereas~$\ell$
        upper-bounds the total number of so-called {\em relocations}.
      For some applications, it might make more sense to bound the number of deletions instead of insertions or the sum of insertions and deletions.
    However, note that while in \cref{def:globalmultistage}~we only upper-bound the number of insertions from $S_i$ to~$S_{i+1}$ explicitly,
		the number of deletions is implicitly upper-bounded by $\ell + \abs{S_1} - \abs{S_\tau}$.  
		Furthermore, for monotone problems removing elements from the solution set is 
		never beneficial except to make space for other elements.
		Thus, any given solution of a monotone problem can be modified to satisfy $\abs{S_1} = \abs{S_2} = \dots = \abs{S_\tau}$,
		so the numbers of insertions and deletions will be equal.

\subparagraph*{Parameterized Algorithmics.}
%A parameterized problem~$\Pi$ is a set of pairs~$(I, k)$, where~$I$ denotes the problem instance and~$k$ is the parameter. 
A parameterized problem~$\mathcal X$ is \emph{fixed-parameter tractable} (FPT) if there exists an algorithm solving any instance~$(I, k)$ of~$\mathcal X$ in~$f(k) \cdot |I|^c$ time, where~$f$ is some computable function and~$c$ is some constant.
If~$\mathcal X$ is shown to be W[1]- or W[2]-hard, then it is presumably not fixed-parameter tractable.
To show W[1]- or W[2]-hardness, one employs \emph{parameterized reductions},
that is, algorithms that map any instance~$(I,k)$
in~$f(k)\cdot |I|^{O(1)}$ time to an equivalent instance~$(I',k')$ with~$k'=g(k)$ for some computable functions~$f,g$.
A \emph{reduction to a problem kernel} is a polynomial-time algorithm that, given an instance~$(I, k)$ of~$\mathcal X$, returns an equivalent instance~$(I', k')$, such that~$|I'| + k' \le g(k)$ for some computable function~$g$. We call~$(I', k')$ a \emph{kernel} of $(I, k)$.
Problem kernels are usually achieved by applying \emph{data reduction rules}.
Given an instance~$(I, k)$, a data reduction rule computes in polynomial time a new instance~$(I', k')$.
We call a data reduction rule \emph{safe} if $(I, k) \in \mathcal X \iff (I', k') \in \mathcal X$.
Clearly, for a decidable problem the existence of a problem kernel implies fixed-parameter tractability.

% For more details on parameterized algorithmics, we refer to Cygan et al.~\cite{CygFKLMPPS15}.

\section{FPT-Frameworks for NP-hard Problems Gone Globally Multistage}\label{sect:fpt}
In this section, we introduce two similar, systematic approaches to show fixed-parameter tractability 
for the combined parameter solution size~$k$ plus number~$\ell$ of relocations, that is~$k+\ell$.
This applies to the global multistage versions of NP-hard problems that 
either are superset-enumerable (see below) or that admit a certain type of kernel and are monotone.  
We exemplify our two frameworks by applying them to the global multistage versions of the NP-hard problems 
\textsc{Cluster Edge Deletion}, \textsc{Cluster Editing}, 
\textsc{Vertex Cover}, and \textsc{Path Contraction}. 

\subsection{Superset-Enumerable Problems}\label{subsect:framework}\label{sect:enum}

Our first framework works for problems with the property
that all solutions containing some given set $F$ and being 
minimal under this condition can be enumerated in FPT time.

\begin{definition}
 A graph property $\mathcal{P_X}$ is \emph{superset-enumerable of size $f$}
for some computable function $f$
if, for any graph on $n$ vertices, any integer $k$, and any set $F$,
the set
\[
	\bigl\{S \supseteq F \bigm| \text{$S$ satisfies $\Pp_\Xx$ in $G$ and $\abs{S} \leq k$}\bigr\}
\]
has at most $f(k)$ minimal elements, and these can be enumerated in $\poly (n) f(k)$ time.
\end{definition}

Let $\mathcal \Pp_\Xx$ be a graph property that is superset-enumerable.
Then our framework solves an instance $(\mathcal G= \langle G_1, \ldots, G_{\tau} \rangle, k, \ell)$ of the \textsc{Global Multistage Problem} for $\mathcal{P_X}$, 
roughly as follows:
Guess a minimal solution in the last layer.
Extend this solution backwards layer by layer as follows.
If the solution $S_i$ ($i = \tau, \ldots, 1$) is also a solution in the previous layer $G_{i-1}$, then set $S_{i-1} := S_i$.
Otherwise, guess a subset $F \subseteq S_i$.
Then enumerate  all minimal solutions containing $F$ in $G_{i-1}$ and guess which of them we should take.
The superset-enumerability ensures that we can compute all these minimal solutions fast enough as well as that their number is bounded.

We now describe the algorithm in detail.

\begin{description}

	\item[Step 1] Guess a sequence $L$ of at most $2\ell+k$ integer pairs $(x,y)$ where 
         $0 \le x \le k$ and $0 \le y \le f(k)$ 
         such that for each pair either $x=0$ or $y=0$ holds.
         %\fk{FK: Ist das n\"otig,
%wir \"uberpr\"ufen doch sowieso, dass $S$ nicht zu gro\ss\ wird:, and $x\neq 0$ holds for at most $\ell$ pairs.}

	      Intuitively, a pair $(x,y) \in L$ with $x\neq 0$ instructs us to delete the $x$th
	      element of the current solution set $S$, and continue with the next pair.
	      If $x=0$, then we enumerate all solutions containing the current solution set $S$, and pick the $y$th such solution as the new current solution (where the elements of $S$ and the solutions are numbered in an arbitrary order).

	\item[Step 2] Start with $S = \emptyset$ and $i = \tau$.
	      Repeat this step as long as  $i \ge 1$.
	      If $S$ does not satisfy $\mathcal P_{\mathcal X}$ on $G_i$,
	      then take (and remove) the first pair of $L$ and modify $S$ accordingly as described above.
	      Otherwise, if $S$ satisfies $\mathcal{P_X}$ on $G_i$,
	      then set $S_i := S$ and decrement $i$.

	\item[Step 3] If at any point we try to take an element from $L$ while it is empty,
              the sum $\sum_{i\in [\tau -1]} |S_{i+1} \setminus S_i|$ exceeds~$\ell$,
              or the size~$|S|$ exceeds~$k$, then restart from Step~1.
	      Otherwise, the sets $S_1,\ldots,S_{\tau}$ form a solution for $(\mathcal G, k, \ell)$ for graph property $\mathcal P_{\mathcal X}$.
\end{description}

\begin{theorem}
Let $\mathcal P_X$ be a graph property that is superset-enumerable of size $f$.
Given an $n$-vertex temporal graph~$\mathcal G = \langle G_1, \ldots, G_{\tau} \rangle$
as well as  integers $k$ and $\ell$, the
\gmp{} $(\mathcal G,k,\ell)$ for~$\mathcal{P_X}$ can be solved in $\mathrm{poly}(n) \tau (k+ f(k) + 1)^{2\ell+k}$ time
and using $O((\ell + k) \log f(k) + k \log n + \log \tau)$ bits in addition to the space needed to 
enumerate and verify the solutions.
%\fk{FK: Ist es max\{...\} oder * ? }
%{\color{red} KH: max, da jede Zahl in der Folge zwischen 0 und $\max\{k, f(k)\}$ ist.}
%\fk{FK: Wir haben $2\ell + k$ Tupel; die erste Komponente hat $k+1$
%M\"oglichkeiten und die zweite $f(k)+1$. Ich komme somit auf $((k+1)\cdot
%(f(k)+1))^{2\ell + k}$ M\"oglichkeiten insgesamt.}
%{\color{red} KH: Du hast recht; es gen\"ugt aber, die Paare mit $x=0$ oder $y=0$ zu betrachten, was zu $(k + f(k) + 1)^{2\ell+k}$ f\"uhrt.}
\label{thm:new_framework}
\end{theorem}
\begin{proof}
Since each element of $L$ can take $k+f(k)$~possible values,
the number of sequences of length at most $2\ell+k$ is bounded by $(k+ f(k))^{2\ell+k}$.
Thus, the running time of our algorithm can be upper-bounded by
$\mathrm{poly}(n) \tau (k+ f(k) + 1)^{2\ell+k}$. 
Storing $L, S$ and $i$ can be done in the space bound described in the
theorem.
The entire sequence $S_1, \dots, S_\tau$ does not need to be stored at any point.
Instead, once a correct guess for $L$ has been found and verified,
$L$~can be used to recompute and immediately output $S_1, \dots, S_\tau$ one after another.

We next show that the algorithm solves the \gmp{}.
Assume that the algorithm returns a solution $(S_1, \dots, S_\tau)$.
Note that our checks in Step 3 guarantee that
% Durch ein Guess kann mehr wie ein Element hinzukommen.
%First note that $\sum_{i\in [\tau -1]} |S_{i+1} \setminus S_i| \le \ell$ as at most $\ell$ pairs $(x, y)$ of a guessed sequence fulfill $x \neq 0$.
%
%Furthermore, the algorithm clearly ensures that the sets $S_i$ computed by the algorithm fulfill $|S_i| \le k$, and that $S_i$ fulfills $\mathcal{X}$.
%Thus, 
$(S_1, \dots, S_\tau)$ is a valid solution.

Now assume that there exists a solution $(S_1, \dots, S_\tau)$ for $(G_1, \ldots, G_\tau)$.
First, note that we may assume that $S_\tau$ is a minimal solution for $G_\tau$, as we only count insertions.
This means that our algorithm finds~$S_\tau$.
Second, note that if $S_i$~fulfills~$\mathcal{X}$ in~$G_{i-1}$ for some~$i\in
\{2,\ldots,\tau\}$, then $(S_1, \dots, S_{i-2}, S_i, S_i, S_{i+1}, \dots, S_\tau)$ is also a solution
since replacing $S_{i-1}$ by $S_i$ increases the total number of insertions by~$|S_{i} \setminus S_{i - 2}| -
(|S_{i} \setminus S_{i- 1}| + |S_{i-1}\setminus S_{i - 2}| ) \le 0$.
Intuitively speaking, we can change elements in the solution as late as
possible while going backwards through the graphs in $(G_1, \ldots, G_\tau)$.

Third note that if an $S^*$ fulfills $\mathcal{P_X}$ in $G_{i-1}$ %is a solution 
with $S_{i - 1} \cap S_i \subseteq S^* \subseteq S_{i - 1}$, % for some $i\in \{2, 3, \dots, \tau \}$, 
then $(S_1, \dots, S_{i -2}, S^* , S_{i}, \dots S_{\tau})$ 
is also a solution, as $ |S_{i} \setminus S^*| + |S^*\setminus S_{i -2} |  \le 
|S_{i}\setminus S_{i-1} | + |S_{i-1} \setminus S_{i-2}|$
follows from the fact that $(S_i \setminus S^*) \cup (S^* \setminus S_{i-2}) \subseteq (S_i \setminus S_{i-1}) \cup (S_{i-1} \setminus S_{i-2})$.
This third point actually generalizes the second observation.

%\fk{Wir zählen nicht DELETIONS, sondern INSERTIONS!!! Wir m\"ussen es von
%hinten kontruieren oder zu Beginn des Beweises sagen, dass wir DELETIONS
%z\"ahlen. Am Ende des Beweisen könnte man dann sagen, um INSERTIONS zu
%z\"ahlen, machen wir es r\"uckw\"arts.}

%{\color{red} KH: Zu insertions ge\"andert.}
To sum up, we may assume that there is a solution $(S_1, \dots, S_\tau)$ such that
\begin{itemize}
\item $S_\tau$ is a minimal solution,
\item if $S_i$ fulfills $\mathcal{P_X}$ in $G_{i - 1}$, then $S_{i - 1 } = S_i$, and 
\item $S_{i-1}$ is minimal amongst all solutions of size at most $k$ containing $S_i \cap S_{i-1}$.
\end{itemize}
Since our algorithm ultimately explores all sequences of this form, we may conclude that it will eventually find a solution if it exists.
\end{proof}

We now show that our framework is applicable to all monotone problems admitting a \emph{full} kernel~\cite{DAMASCHKE2006337},
i.e., a kernel containing every minimal solution.

\begin{proposition}\label{prop:mon-full}
Let $\mathcal{X}$ be a monotone problem admitting a full kernel
of size $f'(k)$, where $k$ is the solution size bound.
Assume that the graph property~$\Pp_\Xx$ can be verified in polynomial time.
Then $\mathcal{P_X}$ is superset-enumerable of size $f(k)=2^{f'(k)}$.
\end{proposition}

\begin{proof}
 Given a graph $G$ and any set $F$, we compute a full kernel $K$ of $G$.
 For each subset $X \subseteq K \setminus F$, we check whether $X \cup F$ satisfies $\Pp_\Xx$, has size at most~$k$,
 and whether $X$ is minimal with these properties (since $\mathcal{X}$ is monotone, checking for minimality can be done by checking for each element $x \in X$ whether $(F\cup X) \setminus \{x\}$ satisfies $\mathcal{P_X}$).
 This clearly runs in \fpt-time and returns at most $2^{f'(k)}$ solutions since $f'(k)$ is an upper bound for the size of the full kernel~$K$.
 Also any solution returned is minimal under all solutions containing~$F$, so it remains to show that every solution~$S$ with this property is returned.
To see this, let $S' \subseteq S$ be a minimal solution (not necessarily containing~$F$).
 We have $S' \setminus F \subseteq S' \subseteq K$ since $K$~is a full kernel.
 As $\Xx$ is monotone, $S' \cup F \subseteq S$ is also a solution, thus $S' \cup F = S$ by minimality of~$S$.
 Thus, when the algorithm tests $X = S' \setminus F \subseteq K \setminus F$, it will output the solution~$S$.
\end{proof}

Combining Theorem~\ref{thm:new_framework} and Proposition~\ref{prop:mon-full}, 
we obtain the following.

\begin{corollary}
Let $\mathcal X$ be a monotone graph problem having a kernel of size $f'(k)$
for the solution size bound~$k$ and some computable function~$f'$.
%such that all minimal solutions of size at most $k$ containing a given set can be enumerated in $f(k)n^c$, and the number of these solutions is bounded by $f(k)$ for a computable function $f$ and a constant $c$.
Given an $n$-vertex temporal graph~$\mathcal G = \langle G_1, \ldots, G_{\tau} \rangle$
as well as  parameters $k$ and $\ell$, the
\gmp{} $(\mathcal G,k,\ell)$ for graph property~$\mathcal{P_X}$ can be solved in
$\mathrm{poly}(n) \tau (\max\{k+1, 2^{f'(k)} + 1\})^{2\ell + k}$ time
and $O((\ell + k) f'(k) + k \log n + \log \tau)$ bits in addition
to the space needed to compute the full kernels and to verify 
the solutions.
\end{corollary}

\subsection{Monotone Problems with Full Kernels}\label{sect:monKernel}

For our second framework, which only targets monotone problems with a full kernel,
we now slightly modify our approach from the last section to improve the
parts of the time bounds that depend on the parameters $k$ and $\ell$. 
This allows us to get improved bounds for 
the global multistage versions of the NP-hard 
problems \textsc{Vertex Cover} and \textsc{Path Contraction}. 
In addition, our second framework allows us to compute kernels for
these two problems.

%\subsection{The Framework}\label{subsect:framework}
Let $\mathcal X$ be a parameterized monotone problem that
admits a {full kernel} 
and that allows for any solution to be verified in time polynomial in the input instance size.
Then, for an instance $(\mathcal G= \langle G_1, \ldots, G_{\tau} \rangle, k, \ell)$ of the \textsc{Global Multistage Problem} for graph property~$\mathcal{P_X}$, 
the rough idea behind our framework is as follows: For each layer~$G_i$ of the temporal graph, we compute a kernel $G_i'$.
Now we can guess the insertions and replacements from  solution $S_1$ to $S_{\tau}$.  
This approach exploits two properties of $\mathcal X$ to make this approach work.
Due to monotonicity, we never have to delete any element except if we aim to replace it,
and the full kernel ensures that each layer contains all minimal solutions of size~$k$ for that layer. 
In the following, we call the reductions that are used to get these instances {\em full reductions}.
We remark that, while a full kernel is required to contain all elements of each minimal solution of size at most~$k$,
our framework can easily be adapted to only require that each element of each minimal solution of size at most~$k$ is uniquely represented by some arbitrary element of the kernel.

%For the following, assume that $\mathcal X$ is a monotone parameterized
%graph problem with parameter $k$ that admits a full kernel of size $f(k)$ and that allows for any solution to be verified in time polynomial in the input instance size. 
%
We next describe our framework for solving the \gmp{} $(\mathcal G, k, \ell)$
for graph property $\mathcal P_{\mathcal X}$ on a temporal graph $\mathcal G = \langle G_1, \ldots, 
G_\tau \rangle$ with lifetime $\tau$.
Even though we can apply the framework of the previous section to full
kernels of size $f(k)$ by enumerating all possible solutions in the kernel,
at every step we have to essentially guess a number in~$1,\ldots,(2^{f(k)} + 1)$. 
In contrast, the revised framework below only requires guesses from~
$1,\ldots,(f(k)+1)$.

To simplify the algorithm and the following proof, we will replace the requirement
$\sum_{i\in [\tau -1 ]} |S_{i+1} \setminus S_{i}| \leq~\ell$ ($*$)
by the modified bound
$\abs{S_1} + \sum_{i\in [\tau -1 ]} |S_{i+1} \setminus S_{i}| \leq~ k + \ell$
($**$)
which is clearly implied by ($*$).
We argue that this does not constitute a restriction
because for any solution that obeys ($**$), 
the algorithm can
%we may 
do the following.
As long as $(*)$ does not hold,
take $i$ as the first time step where $N := S_{i+1} \setminus S_i$ is not empty,
pick an arbitrary element from $N$ and add it to all of $S_1, \dots, S_i$.
It is easy to check that this will yield a solution for ($*$). 

We now give the algorithm for solving the \gmp{}.

\begin{description}
	\item[Step 0] Compute a 
              full kernel $G'_i$ for all $G_i$ ($i \in [\tau]$).
	      In this way, we obtain a sequence of kernels $\mathcal G' = \langle G'_1, \ldots, G'_{\tau} \rangle$. 
	      %(There are several problems where we can add vertices without
	      %changing the solution.  For such problems,
          %   we can turn $\mathcal G'$ easily into a kernel 
          %    (i.e., a temporal graph) by
	      %adding vertices such that each time step has the same vertex set.)

	\item[Step 1] Guess a sequence $L$ of at most $\ell+k$ integer pairs $(x,y)$ where $1 \le y \le f(k)$ as well as $x = 0$ 
		for the first
		$k$ integer pairs and $1 \le x \le f(k)$ for the remaining.
	      Intuitively, a pair $(x,y) \in L$ with $x\neq 0$ instructs us to replace the $x$th
	      element of the current solution set $S$ with the $y$th element of $G'_{i}$. 
	      If $x=0$, then only the $y$th vertex is added.

	\item[Step 2] Start with $S = \emptyset$ and $i = 1$.
	      Repeat this step as long as  $i \leq \tau$.
	      If $S$ does not satisfy $\mathcal P_{\mathcal X}$ on $G'_i$,
	      then take (and remove) the first pair of $L$ and modify $S$ accordingly as described above.
	      Otherwise, if $S$ satisfies~$\mathcal{P_X}$ on $G'_i$,
	      then set $S_i := S$ and increment $i$.

	\item[Step 3] If at any point we try to take an element from~$L$ while it is empty, or if $|S| > k$, then restart from Step~1.
	      Otherwise, the sets $S_1,\ldots,S_{\tau}$ form a solution for $(\mathcal G, k, \ell)$ for graph property $\mathcal P_{\mathcal X}$.
\end{description}

\begin{theorem}
Let $\mathcal X$ be a monotone parameterized graph problem 
such that it admits a full kernel of size $f(k)$ and any solution can be verified in polynomial time.
Given an $n$-vertex temporal graph~$\mathcal G = \langle G_1, \ldots, G_{\tau} \rangle$
as well as  parameters $k$ and $\ell$, the following holds 
for the 
\gmp{} $(\mathcal G,k,\ell)$ for graph property~$\mathcal{P_X}$:
\begin{enumerate}
\item it can be solved in $f(k)^{2\ell + k} \tau \mathrm{poly}(n)$ time and 
$O((\ell + k) \log f(k) + \log \tau)$ bits in addition to the space needed to compute the full kernels and to verify the solutions.
\item if there is a polynomial time computable function $\phi$ which maps a graph 
$G = (V, E)$ and a vertex set $W$ to a graph $\phi(G, W)$ with $V(\phi(G, W)) = V \cup W$
such that $S \subseteq V \cup W$ satisfies~$\mathcal{P_X}$ in $\phi(G, W)$ if and only if $S \cap V$ satisfies $\mathcal{P_X}$ in $G$,
then the \gmp{} has a kernel 
of at most $f(k)\tau$ vertices
 and at most
$(f(k)\tau)^2\tau$~edges.
\end{enumerate}
\label{thm:framework}
\end{theorem}
\begin{proof}
%Let $f(k)$ be an upper bound on the size of the kernel in each time step.
The running time of our algorithm can be upper-bounded by
$f(k)^{2\ell + k} \tau \cdot \mathrm{poly}(n)$
as there are $2\ell + k$ numbers to be guessed from $1,\dots,f(k)$
and the algorithm takes $\tau \cdot \poly(n)$~time for each attempt.
%$(f(k)+1)^{2(\ell+k)} \tau \mathrm{poly}(n)$.
The space bound in the theorem allows us to store
$L$, $S$ and $i$.

It remains to show that
the algorithm finds a solution with bound ($**$) for $\mathcal G$ if it exists. 
%Let us consider a solution
%for $\mathcal G$, which we can turn into a solution for $\mathcal G'$, on 
%which we focus in the following.
\Wilog,  we can assume that the solution for $\mathcal G$ changes the 
vertices/edges
over
time as late as possible.
In particular we first can do all insertions and then all replacements.
By our guess, we determine the vertices of the solution~$S_1$ in~$G_1$. Furthermore,
whenever the vertices in the solution set change, we also can guess
that change in our sequence of pairs. 
Knowing~$S_1$ and the sequence, our algorithm finds the solution for $\mathcal G$.
Finally note that if some instance $(G_i,k)$ is a no-instance, 
then our algorithm finds also no solution.

We now show how to turn the sequence of full kernels $(G_i')_{i \in [\tau]}$
into a kernel for the global multistage problem.
Note that the latter must have the same same vertex set in every time step.

Let $W := \bigcup_{i \in [\tau]} V(G_i')$ and note that $\abs{W} \leq f(k)\tau$.
Then $\langle \phi(G_i', W) \rangle_{i \in [\tau]}$ is a temporal graph
whose number of edges is clearly at most $(f(k)\tau)^2\tau$.
Due to the conditions placed on the function~$\phi$, 
$(\langle \phi(G_i', W) \rangle_{i \in [\tau]} , k , \ell)$ is a kernel for the \gmp{} for~$\Pp_\Xx$.
\end{proof}

The running time of our algorithm can be easily improved 
if the solution set of a monotone graph problem $\mathcal X$ only consists
of vertices---as it is often the case. Then we can 
replace the $f(k)$-term in the running time by the number of 
vertices in the kernel.

\subsection{Applications}\label{subsect:preserving_kernelVC}

\subparagraph*{Cluster Editing and Cluster Edge Deletion.} We start to show
that our framework for superset-enumerable problems can be applied very easily to the
following two clustering problems.

The first problem is \textsc{Cluster Editing} 
 (also known as \textsc{Correlation Clustering})
where the goal is to find a modification of a given graph by
deleting or adding edges such that the modified graph is a {\em cluster
graph}, i.e., all connected components of the 
modified graph are cliques (also called clusters). 

\begin{problem}[\textsc{Global Multistage Cluster Editing}]
        Given a triple $({\mathcal G},k,\ell)$ where
         $\mathcal G = \langle G_1, \ldots, G_{\tau} \rangle$ is a temporal graph,
        the goal is to
        find sets $S_1,\ldots,S_{\tau} \subseteq 
        E(\mathcal{G}) \times \{\operatorname{del}\} \cup 
        \{\{u, v \} : u,v \in V(\mathcal{G})\} 
        \times \{\operatorname{add}\}$, each of size at most $k$, 
        such that $(G_i \setminus \{e \in E(G_i) : (e, \operatorname{del})\in S_i\}) \cup \{ e : (e, \operatorname{add}) \in S_i\}$
        is a cluster graph for $i\in [\tau]$
        and the total number of
        new elements when going from $S_i$ to $S_{i+1}$ ($i \in [\tau-1]$)
        is upper-bounded by~$\ell$, that is, $\sum_{i\in [\tau -1 ]} |S_{i+1} \setminus S_i| \leq~\ell$.
\end{problem}

We mention in passing that \textsc{Cluster Editing} in temporal graphs has also been studied by Chen et al.~\cite{TemporalClusterEditing18}
and a related dynamic model by Luo et al.~\cite{LMNN18}.

The second problem is \textsc{Cluster Edge Deletion}. It is the restriction of \textsc{Cluster Editing} where one is only allowed to delete edges.
\textsc{Global Multistage Cluster Edge Deletion} is the restriction of \textsc{Global Multistage Cluster Editing} with $S_i \subseteq E(\mathcal{G}) \times \{\operatorname{del}\}$ for all $i \in [\tau]$.

%\textsc{Cluster Edge Deletion} and \textsc{Cluster Editing} are {$(3^k, n^3)$-enumerable} (this can be done by marking all edges from the subset $S$ as deleted/added and then applying the $O(3^k +n^3)$-algorithm by B\"ocker et al.~\cite[Theorem 3]{BockerBBT09})
%and thus we can conclude the following.

\begin{corollary}
For an $n$-vertex, $\tau$-layer temporal graph $\Gg$, the
\textsc{Global Multistage Cluster Editing} and \textsc{Global Multistage Cluster Edge Deletion} instances $(\Gg, k, l)$
can both be solved in $1.82^{k(\ell + 1)} \tau n^3$~time. An extra factor of
$(\log k)$ in the running time allows us to solve both problems with 
$O( (\ell+k) k \ell + k\log n + \log \tau)$ bits.
\end{corollary}

\begin{proof}
The claim follows directly from \Cref{thm:new_framework} and %\cite{BockerBBT09}.
the fact that both problems are {$(1.82^k, n^3)$-enumerable}, which can be done by making the modification of an edge from the subset $S$ more expensive than the budget $k$,
and then applying 
% the $O(1.82^k +n^3)$-time algorithm for \textsc{Weighted Cluster Editing} by B\"ocker et al.~\cite[Theorem 7]{BockerBBT09} 
B\"ocker et al.'s search-tree algorithm for \textsc{Weighted Cluster Editing}~\cite[Theorem
7]{BockerBBT09}
(note that this algorithm can also be used for \textsc{Cluster Edge Deletion} by setting the weight of adding an edge to~$k+1$).

The algorithm runs in $O(1.82^k +n^3)$ time, modifies the graph and builds a 
search tree of depth $O(k)$. Each descending step of the search tree
makes $O(1)$ changes to the graph. These changes as well as our modifications due to subset $S$ 
can be maintaines in a heap with
$O(k\log n)$ bits. Evaluating the given graph and the heap, an access to the
modified graph runs in $O(\log k)$ time. By \Cref{thm:new_framework}, the
space bound in the corollary follows.
\end{proof}

%We now briefly discuss the application of our framework to the multistage version of 
%\textsc{Vertex Cover} and \textsc{Path Contraction}.

\subparagraph*{Vertex Cover.}
As a first simple application of our framework for full kernels, we consider \textsc{Global Multistage Vertex
Cover}. 
In a graph $H$, $C \subseteq V(H)$ is a {\em vertex cover} if for every
$\{u, v\}\in
E(H): v \in C$ or $u \in C$. 
In the NP-hard \textsc{Vertex Cover} problem, we are
given $(H, k)$ where $H$ is a graph and $k \in \Nat$. 
The goal is to find a vertex cover of size at most $k$ in $H$.
It is worth noting that \textsc{Vertex Cover} has already been investigated in the classical multistage scenario by Fluschnik et al.~\cite{Fluschnik2019}.

\begin{problem}[\textsc{Global Multistage Vertex Cover}]
        Given a triple $({\mathcal G},k,\ell)$ where 
         $\mathcal G = \langle G_1, \ldots, G_{\tau} \rangle$ is a temporal graph,
        the goal is to
        find sets $S_1,\ldots,S_{\tau}$, each of size at most $k$, such that $S_i \subseteq V(G_i)$
        is a vertex cover for $G_i$ ($i\in [\tau]$) and the total number of
        insertions from $S_i$ to $S_{i+1}$ ($i \in [\tau-1]$)
        is upper-bounded by $\ell$, that is, $\sum_{i\in [\tau -1 ]} |S_{i+1} \setminus S_i| \leq~\ell$.
\end{problem}

For a static $n$-vertex graph $H$, \textsc{Vertex Cover} with parameter $k$ admits
a full kernel of at most $k^2+2k$ vertices and at most $k^2 + k$ edges by applying the following three data reduction
rules (known as Buss kernelization \cite{BG93}) exhaustively~\cite{DAMASCHKE2006337}. %~\cite[p.  21]{CygFKLMPPS15}.
Moreover, Fafianie and Kratsch~\cite{FafK14} have shown that the
kernel can be computed within $O(f(k) \log n)$ bits.
It is not hard to see that a solution can be verified with $O(\log n)$
bits.

\begin{description}
	\item[Rule 1.] Delete all isolated vertices. 
	\item[Rule 2.] If a vertex $v$ has more than $k+1$ incident edges, then delete all
                   except $k+1$~of them. 
	\item[Rule 3.] If Rules 1 and 2 cannot be applied any more and $H$ has more than $k^2 + 2k$ vertices
                 or more than $k^2 + k$ edges,
                   then conclude that the given graph is a no-instance.
\end{description}

Clearly, Rule 1 is a full reduction rule 
since an isolated vertex can never be part of any
minimal solution for~$H$. Rule 2 is also full since a vertex $v$ with $k+1$
edges is part of every solution of size~$k$, i.e., all deleted edges are
covered by $v$. If we apply Rule 3, then one can easily see that we transform a no-instance to a
no-instance. Thus, we get a 
sequence of full kernels and 
we can apply our framework as described in \cref{sect:monKernel}.

To obtain a kernel for an instance $({\mathcal G=\langle G_1, \ldots, G_{\tau} \rangle},k,\ell)$ of \textsc{Global Multistage Vertex Cover},
we can simply apply \cref{thm:framework} (ii) with the function $\phi((V, E), W) := (V \cup W, E)$
since adding isolated vertices does not affect an instance of \textsc{Vertex Cover}.
Note that the resulting kernel contains at most~$(k^2+2k)\tau$ vertices and
at most~$(k^2+k)\tau$ edges.
By \cref{thm:framework}, we can conclude the following.
\begin{corollary}
For an $n$-vertex, $\tau$-layer temporal graph $\Gg$, 
\textsc{Global Multistage Vertex Cover} on an instance $({\mathcal G },k,\ell)$
can be solved 
in $(k^2 + 2k + 1)^{2\ell + k}\tau \cdot \poly($n$)$ time
and $O((\ell + k) \log f(k) + \log \tau + f(k) \log n) = O(f'(k, \ell)\log n + \log \tau)$ bits for some computable function $f'$. 
Moreover, the problem has a
kernel of at most $(k^2+2k)\tau$ vertices and at most $(k^2+k)\tau$ edges 
computable in $O (\tau (n+m))$ time, where $m$ is the number of edges in the underlying graph.
\end{corollary}

\subparagraph*{Path Contraction.}
As a second application of our framework for full kernels, we consider the NP-hard problem \textsc{Path Contraction} \cite{HegHLLP}.
For a graph $H$ and a subset of its edges $C \subseteq E(H)$, 
we write $H/C$ for the graph obtained from~$H$ by contracting each edge in $C$.
(\emph{Contracting} an edge means merging its endpoints into a single vertex and removing any loops or parallel edges afterwards.)
In the \textsc{Path Contraction} problem, we are
given a graph~$H$ and an integer $k \in \Nat$, and the goal is to
find $C \subseteq E(H)$ with $|C| \le k$ such that every connected component of $H/C$ is a path.
Again, the multistage adaption is straightforward.

\begin{problem}[\textsc{Global Multistage Path Contraction}]
        Given a triple $({\mathcal G}, k, \ell)$ where 
         $\mathcal G = \langle G_1, \ldots, G_{\tau} \rangle$ is a temporal graph, 
         the goal is to
   decide whether there exist~$S_1, \ldots, S_{\tau}$ 
   with $S_i \subseteq E(G)$ ($G$ being the underlying graph of $\mathcal G$), 
   each of size at most~$k$, such that $G_i/(S_i \cap E(G_i))$ is a disjoint union of paths for every $i \in [\tau]$,
   and the total number of insertions from~$S_i$ to~$S_{i+1}$ 
   is upper-bounded by $\ell$, that is, $\sum_{i\in [\tau -1 ]} |S_{i+1} \setminus S_i| \leq \ell$.
\end{problem}

\textsc{Path Contraction} has a problem kernel with
at most~$5k + 3$ vertices and at most~$(5k + 3)^2$ edges 
with respect to solution size~$k$~\cite{HegHLLP} by applying the data reduction rules
below on each instance~$(H, k)$ of \textsc{Path Contraction}.
The kernel can be computed within $O(f(k) \log n)$ bits~\cite{KamS20}.
We show that this kernel is also a full kernel.

\begin{description}
	\item[Rule 1.] If any connected component of $H$ contains an edge $e=\{u, v\} \in E(H)$
whose removal disconnects it into two connected components
that contain at least $k + 2$ vertices each, then contract the edge $e$.
	\item[Rule 2.] If Rule 1 is not applicable and any connected component has more than $5k + 3$ vertices, then conclude that
	$(H,k)$ is a no-instance.
\end{description}

Rule~1 is a ``full reduction rule'' since 
every path obtained from $H$ after the contraction of $k$ edges has some vertices
connected to $u$ and some connected to $v$ so that it does not make sense to
contract edge $e$. In other words, every minimal solution of size at most
$k$ must not contain edge~$e$. Rule~2 is surely full.

To obtain a kernel for an instance of \textsc{Global Multistage Path Contraction}, observe as before that adding isolated vertices does not affect the solution.
Hence, the kernel is constructed in the same manner as shown for \textsc{Global Multistage Vertex Cover}.
Thus,  with \cref{thm:framework} we get the following.

\begin{corollary}
For an $n$-vertex, $\tau$-layer temporal graph $\Gg$, 
\textsc{Global Multistage Path Contraction} on an instance $({\mathcal G},k,\ell)$ 
can be solved 
in $(5k + 4)^{4\ell + 2k}\tau \cdot \poly(n)$ time and using $O(f'(k, \ell)\log n + \log \tau)$ bits for some computable function $f'$.
Moreover, the problem has a
kernel of at most $(5k + 3)\tau$ vertices and at most $(5k + 3)^2\tau$  edges computable in $O(\tau n (n+m))$ time, where $m$ is the number of edges in the underlying graph.
\end{corollary}

%Note that there are graph problems such as \textsc{Cluster Edge Deletion}
%(to be defined in Subsection~\ref{sect:ced}) that admit a full kernel but are not monotone and, thus, not suitable for our framework.

\section{Parameterized Hardness for NP-hard Problems Gone Globally
Multistage}\label{sect:w-hardness}
In this section, we explore the limitations with respect to achieving fixed-parameter tractability results of the
global multistage versions of NP-hard
problems. It is clear that the global multistage version of an NP-hard 
problem is NP-hard
for $\tau = 1$ and $\ell = 0$ relocations. We further show that 
the global multistage versions of \textsc{Vertex
Cover}, \textsc{Path Contraction}, and \textsc{Cluster Edge Deletion} are W[1]-hard with respect to solution
size~$k$. For the global multistage versions of \textsc{Planar Dominating Set} and \textsc{Planar Edge Dominating Set}, we even show
W[2]-hardness with respect to the combined parameter solution size~$k$ plus 
number~$\ell$ of relocations.

\subsection{Hardness for the Parameter~$k$}\label{subsect:w1_hardnessVC}
In this subsection, we show that the global multistage versions of \textsc{Vertex Cover}, \textsc{Path Contraction} and \textsc{Cluster Edge Deletion}
become $\wone$-hard with respect to solution size~$k$.
This is each time done by a reduction from \textsc{Clique},
which is well-known to be $\wone$-complete with respect to solution size~\cite{DF13}.
In \textsc{Clique}, given a graph $H$ and a number~$\tilde{k}$, the question is whether~$H$ contains~$\tilde{k}$ pairwise adjacent vertices.
We prove the following. 
\begin{theorem}\label{thm:wone-vc-ced}
 \textsc{Global Multistage Vertex Cover}, \textsc{Global Multistage Path Contraction}, and \textsc{Global Multistage Cluster Edge Deletion}
 are $\wone$-hard parameterized by the solution size~$k$.
\end{theorem}

For the proof of \Cref{thm:wone-vc-ced}, we describe our parameterized reduction from \textsc{Clique} using \textsc{Global Multistage Vertex Cover} in full detail as an example.
We highlight the necessary properties of the corresponding vertex gadget and the edge gadget and prove the correctness of the reduction based on these properties.
Afterwards, we show how to build the vertex gadgets and the edge gadgets for \textsc{Path contraction} and \textsc{Cluster Edge Deletion} based on similar ideas.

Before going into details of the reductions,
we remark that \cref{thm:wone-vc-ced} also implies the \wone-hardness of a variant of the global multistage scenario,
in which the solution for two successive time layers may differ only by $q\ge 1$ elements.
A formal definition of this variant is as follows.

\begin{definition}[\textsc{Global Multistage with $q$-Local Budget} for graph property $\mathcal{P}_X$]
  Given an instance~$(\mathcal G, k, \ell)$ of \textsc{Global Multistage} for graph property $\mathcal{P}_X$,
  the goal is to find a solution $\mathcal{S} = (S_1, \dots, S_\tau)$ to the \textsc{Global Multistage} 
  problem~$(\mathcal G, k, \ell)$
  with the additional restriction that $|S_{i+1} \setminus S_i| \le q$ for all $i\in [\tau-1]$.
\end{definition}

We then obtain the following corollary to \cref{thm:wone-vc-ced}.

\begin{corollary}
  \textsc{Global Multistage Vertex Cover with Local Budget}, \textsc{Global Multistage Path Contraction with Local Budget}, and \textsc{Global Multistage Edge Deletion 
   with $q$-Local Budget} parameterized by $k$ are \wone-hard for any fixed $q\ge 1$.
\end{corollary}

\begin{proof}
%   Assume first that $q=1$.
  Given an instance \textsc{Global Multistage Vertex Cover (Path Contraction, Cluster Edge Deletion)}, 
  add $k$ {\em empty layers}, i.e., $k$ layers with no edges between any two successive layers of the input temporal graph.

  Since any set of vertices is a solution to the added empty layers, all changes to the solution set can be performed element by element in the empty layers.
  Thus, the instances are equivalent.
%
%   To generalize the result from $q=1$ to $q>1$, we can extend the graph in
%   each time step by new connected component that requires exactly $q-1$
%   changes in the solution from one to the next time step.
\end{proof}

\subsubsection{Vertex Cover}
\label{subsect:correctnessVC-hardness}

Let $(H,\tilde{k})$ be an instance of \textsc{Clique}.
Let $V(H) = \{v_1, v_2, \dots, v_n\}$ and $E(H) = \{e_1, e_2, \dots, e_m\}$.
We construct a temporal graph $\mathcal{G}$ as follows.

\subparagraph*{Vertex gadget.} 
For each vertex $v_i\in V(H)$, the graph $\mathcal G$ contains a set $V_i$ of $4\tilde{k}$ vertices $v_i^1, \dots, v_i^{4\tilde{k}}$.
We call this set $V_i$ the \emph{copy set} of $v_i$.

\subparagraph*{Edge gadget.}
An edge gadget $\mathcal G_e$ for an edge $e = \{v_p, v_q\}$ in $H$ adds
$8\tilde{k}^2 n + 1$ layers and the vertices $w_e^1, w_e^2, \dots, w_e^{8\tilde{k}^2n}, w_e^*$ to the temporal graph.
%Each layer contains the same~$8\tilde{k}^2 n + 1$ vertices $w_e^1, w_e^2, \dots, w_e^{8\tilde{k}^2n}, v$.
For $i \in [n]$ and $j\in [8\tilde{k}^2]$, in the $(i + jn)$-th layer of $G_e$, 
vertex $w_e^{i+jn}$ is the only non-isolated
vertex not contained in a copy set. It is connected to all vertices in $V_i$.
In the $(8\tilde{k}^2 n + 1)$-st layer, vertex $w_e^*$ is the only non-isolated
vertex not contained in a copy set.
It is connected to all vertices in~$V_q$ and~$V_p$.
Let $W$ be the set of vertices consisting of all $w_e^i$ and all $w_e^*$ 
over all edges $e\in E(H)$ and $i \in [8\tilde{k}^2n]$.
See \Cref{fig:edgeGadgetVC} for an example.

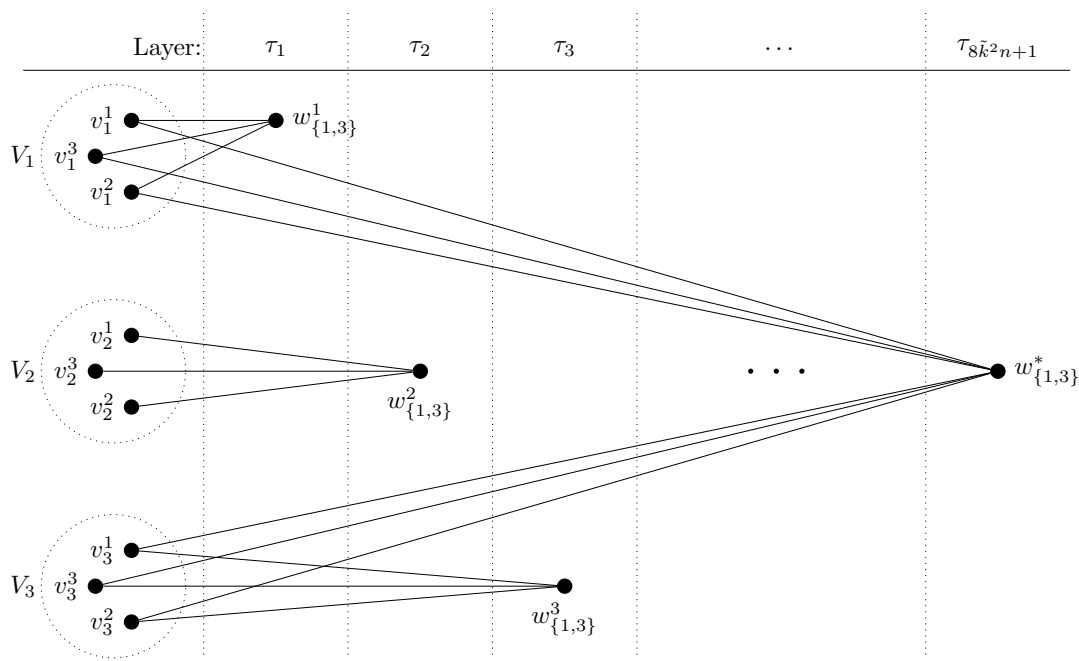
\begin{figure}[t!]
    \scalebox{0.95}{
    \begin{tikzpicture}
      \node[vertex, label=180:$v_1^1$] (v11) at (2,0) {};
      \node[vertex, label=180:$v_1^2$] (v12) at ($(v11) - (0, 1)$) {};
      \node[vertex, label=180:$v_1^3$] (v13) at ($(v11) + (-0.5, -0.5)$) {};
      \draw[dotted] (1.75, -0.5) circle (1);
      \node (c1) at (0.5, -0.5) {$V_1$};
      
      \node[vertex, label=180:$v_2^1$] (v21) at ($(v11) - (0, 3)$) {};
      \node[vertex, label=180:$v_2^2$] (v22) at ($(v21) - (0, 1)$) {};
      \node[vertex, label=180:$v_2^3$] (v23) at ($(v21) + (-0.5, -0.5)$) {};
      \draw[dotted] (1.75, -3.5) circle (1);
      \node (c2) at (0.5, -3.5) {$V_2$};
      
      \node[vertex, label=180:$v_3^1$] (v31) at ($(v21) - (0, 3)$) {};
      \node[vertex, label=180:$v_3^2$] (v32) at ($(v31) - (0, 1)$) {};
      \node[vertex, label=180:$v_3^3$] (v33) at ($(v31) + (-0.5, -0.5)$) {};
      \draw[dotted] (1.75, -6.5) circle (1);
      \node (c3) at (0.5, -6.5) {$V_3$};
      
      \node[vertex, label=0:$w^1_{\{1,3\}}$] (w1) at (4, -0.) {};
      
      \draw (v11) edge (w1);
      \draw (w1) edge (v12);
      \draw (v13) edge (w1);
      
      \node[vertex, label=270:$w^2_{\{1,3\}}$] (w2) at (6, -3.5) {};
      
      \draw (v21) edge (w2);
      \draw (w2) edge (v22);
      \draw (v23) edge (w2);
      
      \node[vertex, label=270:$w^3_{\{1,3\}}$] (w3) at (8, -6.5) {};
      
      \draw (v31) edge (w3);
      \draw (w3) edge (v32);
      \draw (v33) edge (w3);
      
      \node[vertex, label=0:$w^*_{\{1,3\}}$] (ve) at (14, -3.5) {};
      \draw (ve) edge (v11) edge (v12) edge (v13) edge (v31) edge (v32) edge (v33);

      \node (t1) at ($(w1) + (0, 1.)$) {$\tau_1$};
      \node (t2) at ($(w2) + (0, 4.5)$) {$\tau_2$};
      \node (t3) at ($(w3) + (0, 7.5)$) {$\tau_3$};
      
      \node (tl) at ($(ve) + (0, 4.5)$) {$\tau_{8\tilde{k}^2 n + 1}$};
      
      \node (tdots) at ($0.5*(t3) + 0.5*(tl)$) {\dots};
      
      \node (tdots) at ($(tdots) + (0, -4.5)$) {\Huge\dots};
      
      \node (time) at (2.5, 1) {Layer:};
      
      \draw (3, 1.5) edge[dotted] (3, -7.5);
      \draw (5, 1.5) edge[dotted] (5, -7.5);
      \draw (7, 1.5) edge[dotted] (7, -7.5);
      \draw (9, 1.5) edge[dotted] (9, -7.5);
      \draw (13, 1.5) edge[dotted] (13, -7.5);
      
      \draw (0.5, 0.7) edge (15, 0.7);
    \end{tikzpicture}}

  \caption{An example of the edge gadget for the edge $e=\{v_1, v_3\}$ in a graph with three vertices~$v_1$,~$v_2$, and $v_3$ for $\tilde{k} = 3$.
  The $8 \tilde{k}^2$ repetitions of layers 1 to 3 (with different center vertices!) are not drawn to keep the figure readable.
  Furthermore, each copy set contains only three instead of twelve vertices.
  The edges are only contained in the layer in which the right end vertex
  ($w^i_e$ or~$w_e^*$) lies.
  The copy sets are encircled by dotted lines.}
  \label{fig:edgeGadgetVC}
\end{figure}

\subparagraph*{Instance.}
In the first time layer, we have a clique consisting of $4\tilde{k}^2 + 2$
new vertices $V^{\mathrm{first}}$.
The vertices of this clique are isolated in all other layers.
The temporal graph contains $m$ edge gadgets, one for each edge of $H$.
We set~$k: = 4\tilde{k}^2 + 1$ and $\ell = 4\tilde{k}^2 + 8 m \tilde{k}^2 (n - \tilde{k}) + (m - \binom{\tilde{k}}{2})$.
Note that all layers but the first 
contain one star with a center being a vertex $w\in W$ called \emph{center vertex}.

Before we discuss properties of a solution of our \textsc{Global Multistage Vertex Cover} instance, we want to make some basic observation concerning the solution in each layer. 
First, for any~$i\in [n]$ and $j\in [8 \tilde{k}^2]$, in the $(i+jn)$-th layer of an edge gadget for an edge~$e$, the whole copy set $V_i$ is contained in~$S_{i+ jn}$, or~$w^{i+jn}_e$ is contained in $S_{i + jn}$.
Vice versa, both~$V_i$ and~$w_{e^{i+jn}}$ form a vertex cover in this layer.
Each edge gadget contains $8\tilde{k}^2$ of such layers for each vertex~$v_i\in V(H)$.

 Also note that in the last layer of an edge gadget for an edge $e = \{v_p, v_q\}$, the copy sets~$V_p$ and $V_q$ are contained in $S_{8\tilde{k}^2n + 1}$, or $w_e^*$ is contained in $S_{8\tilde{k}^2 n + 1}$. Vice versa, both~$V_p \cup V_q$ and $w_e^*$ are a vertex cover in this layer.

\subparagraph*{Properties of the reduction.}

We start to prove that every solvable instance has a solution
$S_1, \dots, S_\tau$
%the following 
with the following properties. Note that the modifications done to prove one property do
not destroy the each time earlier properties.

\begin{enumerate}

  \item $S_1$ consists of all but one of the $k + 1$ vertices inducing the clique in the first
layer.
  \label{obs:first-layer}\\
  {\bf Proof.} To cover the edges of the clique of size $4\tilde{k}^2 + 2=k+1$ in the first layer, we need~$k$ vertices, i.e., the maximum number of allowed vertices.

  \item For each layer $i$, $|S_i \cap W |\le 1$.
  \label{obs:isolated}\\
{\bf Proof.} By \Cref{obs:first-layer}, we can postpone adding $w\in W$ 
into the solution set until 
$w$ is not isolated in some layer $i$. Since $w$ is isolated in all
subsequent layers $j>i$, we can replace $w$ whenever another vertex $w'\in W$ 
is not isolated and should be added to 
the solution set.

  \item %$|S_1\cap S_i|\le 1$ for all layers $i>1$. 
        $S_2$ contains $\tilde{k}$ copy sets.  \label{obs:further-layer} \\
  {\bf Proof.} The vertices in $S_1$ are isolated in all layers but the first one.
  Thus we may assume that these vertices are relocated before any vertices from $\bigcup_i V_i$.
  Furthermore we may assume that all these relocations already happen in the second layer as postponing them does not gain anything.
  Note that relocating $\abs{V_i}$ vertices from $S_1$ to some copy set $V_i$ initially costs $4\tilde{k}$ relocations,
  but overall saves $8\tilde{k}^2$ relocations inside $W$ in the first edge gadget.
  Thus, each copy set entirely contained in $S_2$ decreases the overall cost of the first edge gadget.
  Since $k = \tilde{k} |V_i| + 1$, this means that $S_2$ contains exactly $\tilde{k}$ copy sets entirely.

   \item For each layer $i$ and each copy set $V_j$: either $|S_i \cap V_j|\le 1$ or $V_j \subseteq S_i$
      \label{obs:completeCS} \\
    {\bf Proof.} Let us consider the smallest $i$ for which the property is
     not true. Assume that only some vertices of the copy set $V_j$ are added to $S_i$.
     Since  $V_j$ is not complete, we can postpone adding the vertices of $V_j$ to the solution set. 
     By induction, copy sets are added only completely to the solution set.

     Assume now that a removal removes vertices from $V_j$ such that more than one vertex of~$V_j$ remains in $S_i$.
     Then this happened since we added a complete copy set $V_p$ into the solution set.
     Thus, there are other copy sets $V_{j_1},\ldots,V_{j_x}$ from which parts were removed from $S_i$. If one of the copy sets $V_{j^*}$ is
     never re-added completely (or if this only happens after a complete removal of $V_{j^*}$), then we can also remove this copy set completely instead of
     removing several copy sets partly.
     Otherwise, let $V_{j^*}$ be the copy
       set that is completely re-added last in a layer $i'>i$.
       If the number of vertices from $W\cup S_1$ that are in the solution
       does not decrease from layer~$i-1$ to layer~$i$,
       then we can completely remove
       $V_{j^*}$ in layer $i$ and re-add it completely in layer $i'$ in such a way that the solution in layer $i'$ is again $S_{i'}$.
       As any copy set which is completely contained in the solution in one layer is again completely contained in the solution of this layer, the solution remains valid. Moreover, it is not hard to see that
       the number of relocations does not increase by this change. (In layer $i$, both solutions perform $4\tilde{k}$ changes.
       The new solution performs $4\tilde{k}$~changes in layer $i'$, while the old solution performs $4\tilde{k}$ changes
       that the new solution does not perform to completely contain $V_{j_1}, \ldots, V_{j_x}$ at some point inside the layers $i+1, \dots, i'$.)

      If the number of vertices from $W\cup S_1$ being in the solution
       decreases from layer~$i-1$ to layer~$i$,
       then we remove all except one vertex of~$V_{j^*}$ together with the vertex from~$W \cup S_1$ and add $V_p$ in layer $i$.
       We add the $4\tilde{k}-1$ vertices of $V_{j^*}$ again in layer $i'$ in such a way that the solution in layer $i'$ is again
       $S_{i'}$. 
       Again, the number of relocations does not increase.  (In layer $i$,
       both solutions perform~$4\tilde{k}$ changes.  The new solution performs
       $4\tilde{k} - 1$ changes in layer $i'$ as the last vertex of $V_{j^*}$
       is not
       relocated between layer $i$ and~$i'$, % (as this case was handled before),
       %and
       %$4\tilde{k} - 1 $ changes otherwise, and 
       but does not perform the $4\tilde{k}-1$ changes to completely contain $V_{j_1}, \dots, V_{j_x}$ at some point
       inside the layers $i+1, \dots, i'$.)

       By applying this change exhaustively, this results in a solution fulfilling the desired property.

   \item $S_i\setminus S_{i-1} \subseteq W$ for any $i\ge 3$.
\label{obs:changes}\\
  {\bf Proof.}
Assume that the property is not true, i.e., some vertices of a complete copy set become part of $S_i\setminus
S_{i-1}$ for some layer $i$ ($i> 2$). Choose $i$ as large as possible such
that a first vertex of a copy set $V_j$ becomes part of the solution, i.e., $S_{i-1} \cap V_j = \emptyset$ and $S_{i} \cap V_j \neq \emptyset$.

We can 
undo the whole relocation of the copy set, which saves $4\tilde{k}$
relocations, but causes some additional relocations:
In the worse
case, we have to pay 
(1) for one more relocation in the current edge gadget (if the newly added copy set has a larger index than the removed
one) and (2) one extra relocation in the last layer of $\tilde{k}-1$ edge
gadgets. To see~(2) observe that
the new copy set can be combined with at most $\tilde{k}-1$ complete other copy
sets in the current solution (\Cref{obs:completeCS})
and so help us to save the extra relocation for $\tilde{k}-1$ 
 edge gadgets.

  \item Every solution for (all layers of) the edge gadget~$\mathcal G_{\{v_i,v_j\}}$ needs at least $8\tilde{k}^2 (n-\tilde{k})$ relocations.
Unless the copy sets $V_i$ and~$V_j$ are in the solution strictly
before the first layer of that gadget begins ({\em good case}),
at least one additional {\em extra relocation} is necessary.
\label{obs:edge-layer}\\
{\bf Proof.}
% Note that exchanging a copy set $V_p$ for a copy set~$V_q$ costs at least~$4\tilde{k} - 1$ relocations (since $|V_q \cap S_i| \le 1$ by \Cref{obs:completeCS}),
% while it can save only two relocations of the center vertex:
% one in the layer where the relocation happens, and one in the last layer of the edge gadget.
By \Cref{obs:changes}, we may assume that only center vertices are relocated.
By choice of $k$, there are at most $\tilde{k}$ copy sets in the solution.
Ignoring the last layer of the edge gadget,
we thus must relocate the $8\tilde{k}^2$ copies of the $n-\tilde{k}$~vertices
whose copy sets are not in the solution.
Furthermore the last layer of the edge gadget needs an extra relocation unless $V_i$ and $V_j$ are part of the solution.
\end{enumerate}

It is easy to see that the reduction is computable in polynomial time. 
We now show the correctness of the two directions of the reduction.

\begin{lemma}\label{lem:W-hardVCforward}
  If $H$ contains a clique of size $\tilde{k}$, then the
\textsc{Global Multistage Vertex Cover} instance
$(\mathcal{G}, k, \ell)$ admits a solution.
\end{lemma}

\begin{proof}
  Let $C= \{v_{i_1}, v_{i_2}, \dots, v_{i_{\tilde{k}}}\}\subseteq V(H)$ be a clique in $H$. 
  We construct a solution $\mathcal{S} = (S_1, \dots, S_\tau)$ for $(\mathcal{G}, k, \ell)$ as follows.
%   In the first layer, $k$ arbitrary vertices of the clique in this layer are contained in the vertex cover~$S_1$.
  In all layers except the first, the vertex cover contains the copy sets~$V_{i_j}$ for all~$j\in [\tilde{k}]$.
  We denote the set of these vertices by $S := \bigcup_{j\in [\tilde{k}]} V_{i_j}$.
%   , the solution contains also the center vertex of this layer, i.e., $S_i := S \cup \{c_{i+ jn}\}$.
  Thus, we have used $k-1$ vertices in our solution.
  The remaining vertex in the solution, called {\em jumping vertex}, %vertex cover 
  is the vertex $c_i$ in all layers $G_i$ in which $S$ is not a vertex cover (i.e., for every edge gadget, $j\in [8\tilde{k}^2]$ and $i\in [n]\setminus \{i_1, \dots, i_{\tilde{k}}\}$, the $(i + jn)$-th layer of the edge gadget, and the last layer of any edge gadget corresponding to an edge $\{v, w\}$ with $v\notin C$).
  Otherwise, we do not change the solution, i.e., $S_i := S_{i-1}$.
  
  The temporal vertex set $\mathcal{S}$ is indeed a vertex cover of size $k$ in each layer.
  There are~$4\tilde{k}^2$~elocations from the first layer to the second layer.
  By \cref{obs:edge-layer}, 
  the jumping vertex %from the vertex cover 
  relocates in each $\mathcal G_e$ with $e$ belonging to the %edge gadget of an edge 
  clique $8\tilde{k}^2 (n-\tilde{k})$ times, while it relocates $8\tilde{k}^2 (n-\tilde{k}) + 1$ 
  times in all other edge gadgets.
  Thus, the total number of relocations is $4\tilde{k}^2 + 8m\tilde{k}^2 (n-\tilde{k}) + m - \binom{\tilde{k}}{2} = \ell$.
\end{proof}

It remains to show the reverse direction.

\begin{lemma}\label{lem:W-hardVCback}
If the
\textsc{Global Multistage Vertex Cover} instance
($\mathcal{G}, k, \ell)$ admits a solution, then 
$H$ contains a clique of size $\tilde{k}$.
\end{lemma}

\begin{proof}
Let us consider a solution that adds $\tilde{k}$ copy sets in $S_2\setminus
S_1$ and adds no further copy sets in later layers 
(\Cref{obs:further-layer,obs:changes}). 
Furthermore, the edge gadget causes in total $8m\tilde{k}^2 (n-\tilde{k})$
relocations by \Cref{obs:edge-layer}.
 Thus, $m - \binom{\tilde{k}}{2}$ relocations
remain from the $\ell$ allowed relocations. 
Thus there are $\binom{\tilde{k}}{2}$ edge gadgets for which we
do not pay an extra relocation due to the last layer. To save the extra relocation
for that many edge gadgets, we must have an edge between all pairs of vertices whose 
copy sets are in the solution. To sum up, there is a clique of size
$\tilde{k}$.
\end{proof}

\subsubsection{Path Contraction}\label{sect:hardness_path_contraction}

\subparagraph*{Vertex gadget.}

For each $v_i \in V(H)$, we add a path of length $4 \tilde{k}$, and call the endpoints of the path $s_i$ and $t_i$.
The edges of this path are the \emph{copy set} $V_i$ of $v_i$.

\subparagraph*{Edge gadget.}

An edge gadget for an edge $e = \{v_p, v_q\}$ adds $8 \tilde{k}^2n + 1$ layers.
It adds $16 \tilde{k}^2 n + 4$ vertices $w_e^1$, $w_e^2, \dots, w_e^{8\tilde{k}^2 n}$, $x_e^1, \dots, x_e^{8\tilde{k}^2n}$, $y_e^1, y_e^2$, and $z_e^1 $ and $z_e^2$, where $w_e^r$ and $x_e^{r}$ 
($r\in [8\tilde{k}^2 n]$)
are isolated in all but one layer of the gadget.
For each $i\in [n]$ and $j\in [8\tilde{k}^2]$, the edges $\{w_e^{i + jn}, s_i\}$ and $\{s_i, x_e^{i+jn}\}$ are present in the $(i + jn)$-th layer of the gadget.
In the $(8\tilde{k}^2n + 1)$-th layer, the edges $\{t_p, t_q\}$, $\{s_p, y_e^1 \}$, $\{y_e^1, y_e^2\}$, $\{s_p,z_e^1\}$, and $\{z_e^1, z_e^2\}$ are present.
Let $W$ be the set of edges incident to $w_e^i$, $x_e^i$, $y_e^1$, or $z^1_e$
for some edge $e$ and $i \in [8\tilde{k}^2n]$.
See \Cref{fig:edgeGadgetPC} for an example.

The edges $\{w^{i+jn}, s_i\}$ and $\{\{s_p, y_e^1\}, \{y_e^1, y_e^2\}\}$ now correspond to the center vertices from the reduction for \textsc{Global Multistage Vertex Cover} (meaning that these edges will be added to the solution using relocations for layers not corresponding to edges or vertices of the clique).

\begin{figure}[t!]
    \scalebox{0.95}{
    \begin{tikzpicture}
      \node[vertex, label=180:$s_1$] (v11) at (2,0) {};
%       \node[vertex] (v12) at ($(v11) - (-0.5, 0.5)$) {};
      \node[vertex, label=180:$t_1$] (v13) at ($(v11) + (-0, -1)$) {};
      \draw[dotted] (1.75, -0.5) circle (1);
      \node (c1) at (0.5, -0.5) {$V_1$};
      \draw (v11) edge[dotted] (v13);

      \node[vertex, label=180:$s_2$] (v21) at ($(v11) - (0, 3)$) {};
%       \node[vertex] (v22) at ($(v21) - (-0.5, 0.5)$) {};
      \node[vertex, label=180:$t_2$] (v23) at ($(v21) + (-0, -1)$) {};
      \draw[dotted] (1.75, -3.5) circle (1);
      \node (c2) at (0.5, -3.5) {$V_2$};
      \draw (v21) edge[dotted] (v23);

      \node[vertex, label=180:$s_3$] (v31) at ($(v21) - (0, 3)$) {};
%       \node[vertex] (v32) at ($(v31) - (-0.5, 0.5)$) {};
      \node[vertex, label=180:$t_3$] (v33) at ($(v31) + (0, -1)$) {};
      \draw[dotted] (1.75, -6.5) circle (1);
      \node (c3) at (0.5, -6.5) {$V_3$};

      \draw (v31) edge[dotted] (v33);

      \node[vertex, label=0:$w^1_{\{1,3\}}$] (w1) at (4, -0.) {};
      \node[vertex, label=270:$x^1_{\{1,3\}}$] (x1) at ($(w1) + (0, -1.1)$) {};

      \draw (v11) edge (w1);
      \draw (x1) edge (v11);
%       \draw[green] (v13) edge (x1);

      \node[vertex, label=90:$w^2_{\{1,3\}}$] (w2) at (6, -3.5) {};
      \node[vertex, label=270:$x^2_{\{1,3\}}$] (x2) at ($(w2) + (0, -1)$) {};

      \draw (v21) edge (w2);
      \draw (x2) edge (v21);
%       \draw[green] (v23) edge (x2);

      \node[vertex, label=90:$w^3_{\{1,3\}}$] (w3) at (8, -6.5) {};
      \node[vertex, label=270:$x^3_{\{1,3\}}$] (x3) at ($(w3) + (0, -1)$) {};

      \draw (v31) edge (w3);
      \draw (x3) edge (v31);
%       \draw[green] (v33) edge (x3);

      \node[vertex, label=270:$y_{\{1,3\}}^1$] (ve) at (12, -2.5) {};
      \node[vertex, label=0:$y_{\{1,3\}}^2$] (y2) at (14, -2.5) {};
      \node[vertex, label=270:$z_{\{1,3\}}^1$] (z1) at (12, -4.5) {};
      \node[vertex, label=0:$z_{\{1,3\}}^2$] (z2) at (14, -4.5) {};
%       \draw (ve) edge (v11) edge (v12) edge (v13) edge (v31) edge (v32) edge (v33);

      \draw (z2) -- (z1) -- (v11) -- (ve) -- (y2);

      \draw (v33) edge[dashed, bend left=35] (v13);

      \node (t1) at ($(w1) + (0, 1.)$) {$\tau_1$};
      \node (t2) at ($(w2) + (0, 4.5)$) {$\tau_2 $};
      \node (t3) at ($(w3) + (0, 7.5)$) {$\tau_3 $};

      \node (tl) at ($(ve) + (1, 3.5)$) {$\tau_{8\tilde{k}^2 n + 1} $};

      \node (tdots) at ($0.5*(t3) + 0.5*(tl) - (0.5, 1)$) {\dots};

      \node (tdots) at ($(tdots) + (0, -4.5)$) {\Huge\dots};

      \node (time) at (2.5, 1) {Layer:};

      \draw (3, 1.5) edge[dotted] (3, -7.5);
      \draw (5, 1.5) edge[dotted] (5, -7.5);
      \draw (7, 1.5) edge[dotted] (7, -7.5);
      \draw (9, 1.5) edge[dotted] (9, -7.5);
      \draw (11, 1.5) edge[dotted] (11, -7.5);
%       \draw (13, 1.5) edge[dotted] (13, -7.5);

      \draw (0.5, 0.7) edge (15, 0.7);
    \end{tikzpicture}}

  \caption{An example of the edge gadget for the edge $e=\{v_1, v_3\}$ in a graph with three vertices~$v_1$,~$v_2$, and $v_3$ for $\tilde{k} = 3$.
  The $8 \tilde{k}^2$ repetitions of layers 1 to 3 (with different $x_e^i$ and $w_e^i$) are not drawn to keep the figure readable.
  The dotted edges represent a path of length $4\tilde{k}^2$ (the copy set of the corresponding vertex).
  The edges are only contained in the layers in which the right end vertex
  ($w^i_e$ or~$v_e$) lies, with the exception of the edges inside a copy set, which are contained in all layers.
  The dashed edge is contained only in the last layer ($\tau_{8\tilde{k}^2n +1}$).
  The copy sets are encircled by dotted lines.}
  \label{fig:edgeGadgetPC}
\end{figure}
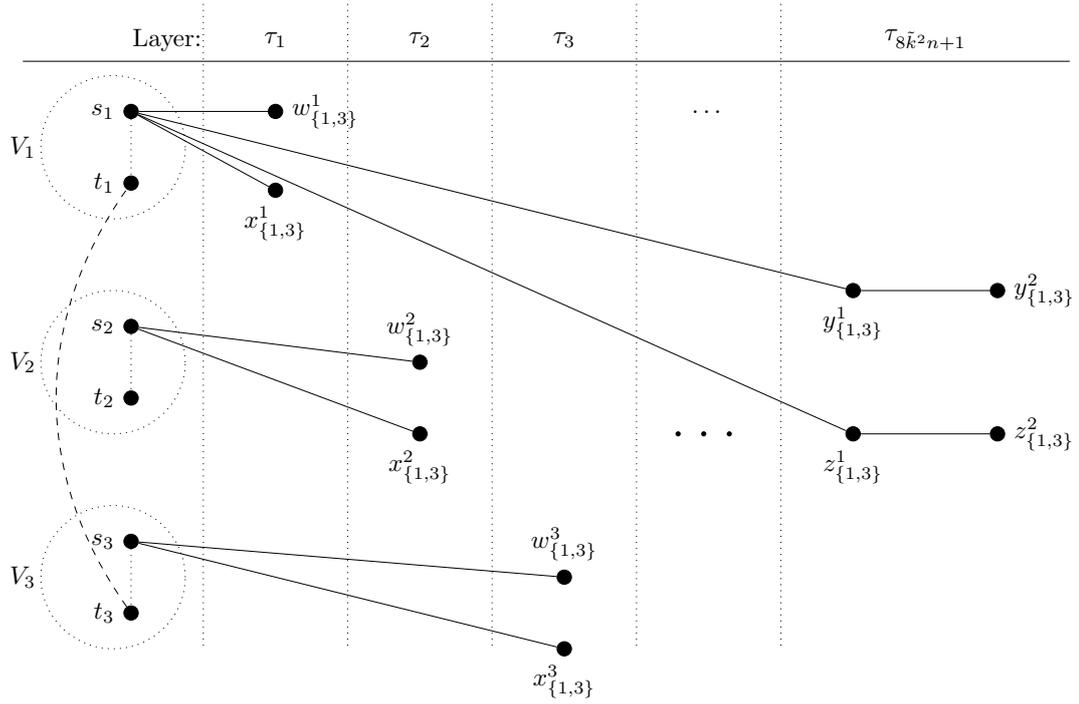

\subparagraph*{Instance.}
Before all edge gadgets, in a first layer we have a clique of size $4\tilde k^2 + 4$, ensuring that the solution in the first layer consists of $4 \tilde k^2 + 2$ edges of these paths.
The vertices of this clique are isolated in all other layers.
Afterwards, there are $m$~edge gadgets, one for each edge of the graph~$H$.
We set $k = 4 \tilde k^2  + 2$ and~$\ell = 4 \tilde k^2 + 8m \tilde k^2 (n - \tilde k) + (2m - \binom{\tilde k}{2})$.

Note that this reduction and the previous reduction are very similar.
\Cref{obs:first-layer}--\Cref{obs:changes} are valid also for the current
reduction (with $|S_i \cap W| \le 2$ in \Cref{obs:isolated} and $|S_1 \cap S_i|\le 2$ in \Cref{obs:completeCS}). Moreover, \Cref{obs:edge-layer} for an edge gadget of edge $\{i,j\}$ is
also 
valid if we 
%. We only have to
increment the number of relocations for the good case by one from $8\tilde{k}^2 (n-\tilde{k})$
to $8\tilde{k}^2 (n-\tilde{k})+1$. As before, if we are not in the good
case, then we have one extra relocation since
the last layer cause an extra relocation exactly if one of the
copy sets $V_i$ and $V_j$ is not in the solution set. To see this, note
that the last layer of the gadget forms a star whose center is the endpoint
of three paths.  Two are of length two and one is either very long 
(if $V_i$ or~$V_j$ is not in the solution set) or of length 1.
To find a solution, a shortest path has to be contracted.

It is easy to see that the reduction is computable in polynomial time. 
By replacing ``vertex cover'' through ``path contraction'' and $8\tilde{k}^2 (n-\tilde{k})$
by $8\tilde{k}^2 (n-\tilde{k})+1$ in the proofs of Lemmas~\ref{lem:W-hardVCforward}
and~\ref{lem:W-hardVCback}, we obtain a proof that shows the following:
 $H$ contains a clique of size $\tilde{k}$ exactly if the
\textsc{Global Multistage Path Contraction} instance
$(\mathcal{G}, k, \ell)$ admits a solution.

\subsubsection{Cluster Edge Deletion}\label{sect:ced}

We now show $\wone$-hardness with respect to solution size~$k$ for \textsc{Global Multistage Cluster Edge Deletion}.
Recall that, for a graph $H$, a set $D \subseteq E(H)$ is called a \emph{cluster edge deletion set} if $H-D$ is a cluster graph (i.e., all connected components are cliques).
%\begin{problem}[\textsc{Global Multistage Cluster Edge Deletion}]
%        Given a triple $({\mathcal G},k,\ell)$ where
%         $\mathcal G = \langle G_1, \ldots, G_{\tau} \rangle$ is a temporal graph,
%        the goal is to
%        find sets $S_1,\ldots,S_{\tau}$, each of size at most $k$, such that $S_i \cap E(G_i) \subseteq E(G)$
%        is a cluster edge deletion set for $G_i$ ($i\in [\tau]$)
%        and the total number of
%        new elements when going from $S_i$ to $S_{i+1}$ ($i \in [\tau-1]$)
%        is upper-bounded by~$\ell$, that is, $\sum_{i\in [\tau -1 ]} |S_{i+1} \setminus S_i| \leq~\ell$.
%\end{problem}
We again reduce from \textsc{Clique}. Let $(H, \tilde k)$ be an instance of \textsc{Clique}. We construct a temporal graph~$\mathcal G$ as follows.

\subparagraph*{Vertex gadget.}
For each vertex $v_i \in V(H)$, we add a clique of size $4 \tilde k + 1$ containing a special vertex~$u_i$.
The edges adjacent to $u_i$ in this clique are the \emph{copy set} of $v_i$.

\subparagraph*{Edge gadget.}
An edge gadget for an edge $e= \{v_p, v_q\}$ adds $8 \tilde k^2 n + 1$ layers.
It adds~$8\tilde{k}^2 n$~vertices $w_e^1, w_e^2, \dots, w_e^{8\tilde{k}^2n}$, which are all isolated in all but one layer.
In the $(i + jn)$-th ($j\in [8\tilde{k}^2]$) layer of the gadget, vertex $w_e^{i+jn}$ is connected to $u_i$.
In the $(8\tilde{k}^2 n + 1)$-th layer, the edge $\{u_p, u_q\}$ is present.
See \Cref{fig:edgeGadgetCED} for an example.

The edges $\{w_e^{i + jn}, u_i\}$ and $\{u_p, u_q\}$ correspond to the center vertices from the reduction for \textsc{Global Multistage Vertex Cover}.

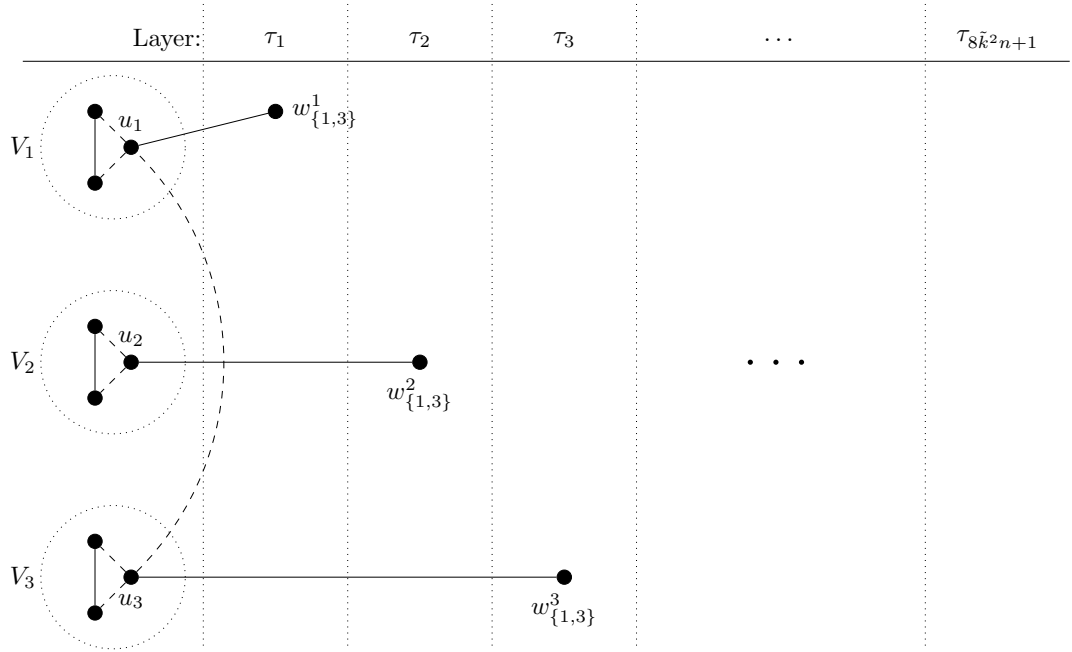
\begin{figure}[t!]
    \scalebox{0.95}{
    \begin{tikzpicture}
      \node[vertex] (v11) at (1.5,0) {};
      \node[vertex] (v12) at ($(v11) - (0, 1)$) {};
      \node[vertex, label=90:$u_1$] (v13) at ($(v11) + (0.5, -0.5)$) {};
      \draw[dotted] (1.75, -0.5) circle (1);
      \node (c1) at (0.5, -0.5) {$V_1$};

      \draw (v11) edge (v12);
      \draw (v13) edge[dashed] (v12) edge[dashed] (v11);

      \node[vertex] (v21) at ($(v11) - (0, 3)$) {};
      \node[vertex] (v22) at ($(v21) - (0, 1)$) {};
      \node[vertex, label=90:$u_2$] (v23) at ($(v21) + (0.5, -0.5)$) {};
      \draw[dotted] (1.75, -3.5) circle (1);
      \node (c2) at (0.5, -3.5) {$V_2$};

      \draw (v21) edge (v22);
      \draw (v23) edge[dashed] (v22) edge[dashed] (v21);

      \node[vertex] (v31) at ($(v21) - (0, 3)$) {};
      \node[vertex] (v32) at ($(v31) - (0, 1)$) {};
      \node[vertex, label=270:$u_3$] (v33) at ($(v31) + (0.5, -0.5)$) {};
      \draw[dotted] (1.75, -6.5) circle (1);
      \node (c3) at (0.5, -6.5) {$V_3$};

      \draw (v31) edge (v32);
      \draw (v33) edge[dashed] (v32) edge[dashed] (v31);

      \node[vertex, label=0:$w^1_{\{1,3\}}$] (w1) at (4, -0.) {};

      \draw (v13) edge (w1);

      \node[vertex, label=270:$w^2_{\{1,3\}}$] (w2) at (6, -3.5) {};

      \draw (v23) edge (w2);

      \node[vertex, label=270:$w^3_{\{1,3\}}$] (w3) at (8, -6.5) {};

      \draw (v33) edge (w3);

      \node (ve) at (14, -3.5) {};
%       \draw (ve) edge (v13) edge (v33);

      \draw (v33) edge[dashed, bend right=45] (v13);

      \node (t1) at ($(w1) + (0, 1.)$) {$\tau_1$};
      \node (t2) at ($(w2) + (0, 4.5)$) {$\tau_2$};
      \node (t3) at ($(w3) + (0, 7.5)$) {$\tau_3$};

      \node (tl) at ($(ve) + (0, 4.5)$) {{$\tau_{8\tilde{k}^2 n + 1}$}};

      \node (tdots) at ($0.5*(t3) + 0.5*(tl)$) {\dots};

      \node (tdots) at ($(tdots) + (0, -4.5)$) {\Huge\dots};

      \node (time) at (2.5, 1) {Layer:};

      \draw (3, 1.5) edge[dotted] (3, -7.5);
      \draw (5, 1.5) edge[dotted] (5, -7.5);
      \draw (7, 1.5) edge[dotted] (7, -7.5);
      \draw (9, 1.5) edge[dotted] (9, -7.5);
      \draw (13, 1.5) edge[dotted] (13, -7.5);

      \draw (0.5, 0.7) edge (15, 0.7);
    \end{tikzpicture}}

  \caption{An example of the edge gadget for the edge $e=\{v_1, v_3\}$ in a graph with three vertices~$v_1$,~$v_2$, and $v_3$ for $\tilde{k} = 3$.
  The $8 \tilde{k}^2$ repetitions of layers 1 to 3 (with different center vertices!) are not drawn to keep the figure readable.
  Furthermore, each copy set contains only two instead of twelve edges.
  The edges are only contained in the layer in which the right end vertex
  ($w^i_e$ or~$w^*_e$) lies, except for the long dashed edge, which is only present in layer $8\tilde{k}^2 n + 1$.
  The copy set $V_i$ consists of the short dashed edges contained in the corresponding dotted circle.}
  \label{fig:edgeGadgetCED}
\end{figure}

\subparagraph*{Instance.}
In the first layer of the temporal graph, we have $4 \tilde k^2 + 1$ vertex-disjoint paths of length two, ensuring that the solution in the first layer consists of $4 \tilde k^2 + 1$ edges of these paths.
The vertices of the paths are isolated in all other layers.
Afterwards, there are $m$~edge gadgets, one for each edge of graph $H$.
We set $k = 4 \tilde k^2  + 1$ and~$\ell = 8m \tilde k^2 (n - \tilde k) + 4\tilde k^2 + ( m - \binom{\tilde k}{2})$.

Note that this and the reduction for \textsc{Global Multistage Vertex Cover} are very similar.
\Cref{obs:first-layer}--\Cref{obs:changes} are valid also for the current
reduction.
The proof of \Cref{obs:completeCS} is even simpler, as any solution either contains no edge of a copy set for a vertex~$v_i$, or it separates $u_i$ from the rest of the clique.
This separation can be done without loss of generality by deleting the edges from the copy set.
Moreover, \Cref{obs:edge-layer} is
also
valid for all edge gadgets. As before, if we are not in the good
case, we have one extra relocation since
the last layer causes an extra relocation exactly if one of the
copy sets $V_i$ and $V_j$ is not in the solution set. To see this, note
that the last layer of the gadget forms two cliques connected by the edge $\{u_p, u_q\}$.
To find a solution, the edge $\{u_p, u_q\}$ has to be deleted, or the vertices $u_p$ and $u_q$ have to be separated from the clique they are contained in, which needs at least $4\tilde{k}$ edges and can be done by deleting the copy set of these vertices.

It is easy to see that the reduction is computable in polynomial time.
Replacing ``vertex cover'' by ``cluster edge deletion set'' in the proofs of Lemmas~\ref{lem:W-hardVCforward}
and~\ref{lem:W-hardVCback}, we obtain a proof showing that
 $H$ contains a clique of size $\tilde{k}$ exactly if the
\textsc{Global Multistage Cluster Edge Deletion} instance
$(\mathcal{G}, k, \ell)$ admits a solution.

We remark that the reduction for \textsc{Global Multistage Cluster Edge Deletion} also yields hardness for
\textsc{Global Multistage Cluster Editing}.

%the global multistage version of the frequently 
%studied graph clustering problem \textsc{Cluster Editing} (also known as \textsc{Correlation Clustering}), where one allows both edge deletions and edge insertions.
%
%\begin{problem}[\textsc{Global Multistage Cluster Editing}]
%        Given a triple $({\mathcal G},k,\ell)$ where
%         $\mathcal G = \langle G_1, \ldots, G_{\tau} \rangle$ is a temporal graph,
%        the goal is to
%        find sets $S_1,\ldots,S_{\tau}$ and $T_1, \ldots, T_\tau$ with $|S_i| + |T_i| \le k$, such that $H_i := (V, \bigl( E(G_i) \setminus S_i\bigr) \cup T_i)$
%        is a cluster graph for all $i\in [\tau]$
%        and the total number of
%        new elements when going from $S_i$ to $S_{i+1}$ and $T_i$ to $T_{i+1}$ ($i \in [\tau-1]$)
%        is upper-bounded by~$\ell$, that is, $\sum_{i\in [\tau -1 ]} |S_{i+1} \setminus S_i| + |T_{i+1 } \setminus T_i| \leq~\ell$.
%\end{problem}

\subsection{Hardness for the Combined Parameter~$k+\ell$}\label{sec:hardPDS}
The problems investigated in the previous section are all superset-enumerable and, therefore, their global multistage variants admit \fpt-algorithms parameterized by $k + \ell$.
However, for some classical, monotone graph problems no superset-enumeration by solution size is known.
Two notable examples are \textsc{Planar Dominating Set} and \textsc{Planar Edge Dominating Set}, which in their classic (static) variants are known to be fixed-parameter tractable for the 
parameter solution size.
Assuming $\fpt{} \neq \wtwo$, we will that these problems are not superset-enumerable
and do not admit FPT algorithms for the their global multistage versions
by proving that these are both $\wtwo$-hard for parameter $k $ even if $\ell= 0$.
Note that this also implies \wone-/\wtwo-hardness for the classical multistage version of these problems.

\subsubsection{Planar Dominating Set}
Recall that, for a graph $H$, a set $D \subseteq V(H)$ is called a \emph{dominating set}
if $D \cup \bigcup_{d \in D} N(d) = V(H)$.
\begin{problem}[\textsc{Global Multistage Planar Dominating Set}]
        Given a temporal graph $\mathcal G = \langle G_1, \ldots, G_{\tau} \rangle$
        that is planar in every time step
         and parameters $k,\ell$,
        find sets $S_1,\ldots,S_{\tau}$, each of size at most $k$, such that $S_i \subseteq V(G_i)$
        is a dominating set for $G_i$ ($i\in \{1,\ldots,\tau\}$) and the total number of insertions
        $\sum_i \abs{S_{i+1} \setminus S_i}$
        is at most $\ell$.
\end{problem}
\begin{figure}[t]
  \begin{subfigure}[t]{0.4\textwidth}
    \centering
      \begin{tikzpicture}[xscale = 1.5, yscale =0.5, baseline={(0,-0.1)}]
  		\foreach \x in {1,...,5}{%
    		\pgfmathparse{(\x-0.8)*65}
    		\node[vertex, fill=lightgray] (C\x) 
    		   at ($(\pgfmathresult:0.6cm) + (1.2cm, 3.5cm)$) {};
  		} 
  		\foreach \x [count=\xi from 1] in {2,...,5}{%
    		\foreach \y in {\x,...,5}{%
        		\draw (C\xi) -- (C\y);
  			}
		}
      
        \foreach \x in {1,2,3} {
        	\node[vertex] (v\x) at ( 0.2, \x * 1.2) {};
        }
        \node[vertex, left= 0.5cm of v2, label=180:$v_i$] (vi) {};
        \node (h) at ($(vi) + (1cm, 3cm)$) {$H$};
        \foreach \x in {1,2,3} {
        	\draw (vi) -- (v\x);
        }
        
        \node[vertex, fill=lightgray] (n1) at (0.6, 1) {};
        \node[vertex, fill=lightgray] (n4) at (1.2, 2) {};
        \node[vertex, fill=lightgray] (n5) at (1.2, 1) {};
        \node[vertex, fill=lightgray] (n6) at (0.6, 2) {};
        
        \foreach \x/\y in {4/6, 6/5, 6/1, 1/5} {
        	\draw (n\x) -- (n\y);
        }
        
        \draw (v2) -- (v3)
        (v2) -- (C5)
        (v1) -- (n4)
        (v1) -- (n6)
        (v2) -- (n1)
        (v2) -- (C4)
        (v3) -- (C3);
        
    \end{tikzpicture}
  \end{subfigure}\hfill
  \begin{subfigure}[t]{0.6\textwidth}
    \centering
	\begin{tikzpicture}[xscale = 1.5, yscale =0.5]
        \foreach \x in {1,2,4,5,6,9} {
        	\node[vertex, fill=lightgray] (v\x) at (\x * 0.3, 0.2) {};
        }
        \node[right=0 of v9] (dotsV) {$\cdots$};
        \node[vertex, right=0 of dotsV, fill=lightgray] (v10) {};
        
        \foreach \x in {3,7,8} {
        	\node[vertex] (v\x) at (\x * 0.3, 0.2) {};
        }
        
        \node[vertex, below= 0.5cm of v6, label=183:$v_i$] (vi) {};
        \node[vertex, above= 0.5cm of v6, label=170:$w$, fill=lightgray] (w) {};
        \node[vertex, right= 0.5cm of w, label=13:$w'$, fill=lightgray] (w1) {};
        
        \foreach \x in {3,7,8} {
        	\draw (vi) -- (v\x);
        }
        
        \foreach \x in {1,2,...,10} {
        	\draw (w) -- (v\x);
        }
        
        \draw (w) -- (w1);
        
        \node[above=0.5cm of w] (gi) {$G_i$};
        \draw [dashed] ($(gi) - (1.7cm, -0.3cm)$) -- ($(gi) - (1.7cm, 5.5cm)$);
        \draw [dashed] ($(gi) - (-1.7cm, -0.3cm)$) -- ($(gi) - (-1.7cm, 5.5cm)$);
		\node (dotsL) at ($(gi) - (1.9cm, 3cm)$) {$\cdots$};
		\node (dotsR) at ($(gi) - (-1.95cm, 3cm)$) {$\cdots$};
	\end{tikzpicture}
  \end{subfigure}
  \caption{For each vertex $v_i$ in the graph $H$, a graph $G_i$ is constructed where vertex $w$ is connected to all vertices except for $v_i$, and~$v_i$ is connected only to its neighbors in $H$.}\label{fig:examplePDS}
\end{figure}
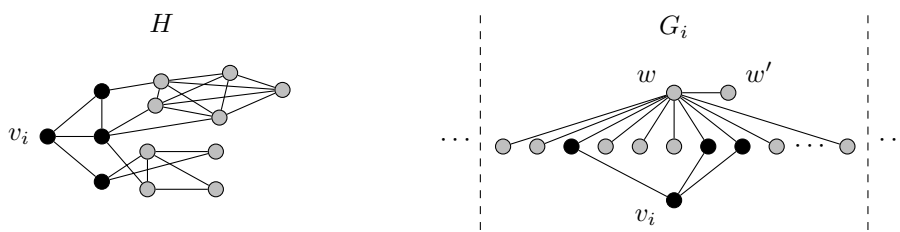
\begin{proposition}
\textsc{Global Multistage Planar Dominating Set} is $\wtwo$-hard when parameterized by the solution size $k$ even when no relocations are allowed $(\ell = 0)$.
\end{proposition}

\begin{proof}
We provide a parameterized reduction from the $\wtwo$-hard \textsc{Dominating Set} problem~\cite{DF13}.
Given a \textsc{Dominating Set} instance $\mathcal{I} = (H, k')$, 
we construct an instance $\mathcal{I}' = (\mathcal G, k, \ell = 0)$ of \textsc{Global Multistage Planar Dominating set} as follows.

Let $w$ and $w'$ be two new vertices.
For each vertex $v_i \in V(H)$ ($i = 1, \ldots, n$) create a 
time layer $G_i$ with
$V(G_i) = V(H) \cup \{w, w'\}$ 
and 
$E(G_i) = \{ \{v_i, u\} |\ \forall u \in N(v_i)\}
\cup \{ \{w, u\}\ |\ \forall u \in V(H) \setminus \{v_i\} \}
\cup \{\{w, w'\} \}$.
For an example see \cref{fig:examplePDS}.
Observe that in every time step $i \in \{ 1, \ldots, n\}$, $G_i$
without $v_i$ forms a star and connecting $v_i$ with some vertices of the
star preserves planarity.

We claim that $H$ contains a dominating set $D$ of size $k$ if and only if 
$\mathcal G$ contains a multistage dominating set $D'$ of size $k+1$.
For the forward direction, choose $D'=D\cup \{w\}$. For every time step $i$, vertex $v_i$ is dominated
in $G_i$
by a vertex in $D$. All other vertices are dominated in $G_i$ by vertex $w$.
For the reverse direction, we have to dominate $w'$ in each time step.
\Wilog{}, we can assume that $w\in D'$. 
Then $D\setminus \{w'\}$ must be a dominating set in $H$ since
every vertex $v_i$ is, at time step $i$, only adjacent to its neighbors from $H$.
\end{proof}

\subsubsection{Planar Edge Dominating Set}
For a graph $H$, $D \subseteq E(H)$ is called an \emph{edge dominating set} if
each edge of $H$ is incident to an edge in $D$.
\textsc{Edge Dominating Set}, the problem of finding a minimum-size edge dominating set, is closely related to \textsc{Vertex Cover} and
\textsc{Matching}.
However, even when restricting \textsc{Edge Dominating Set} to planar 
graphs, within the global multistage setting it appears to be 
computationally harder than \textsc{Vertex Cover} on general graphs, as we will show now.
The restriction of \textsc{Global Multistage Edge Dominating Set}, which we call \textsc{Global Multistage Fully Planar Edge Dominating Set}, is stronger than the one of \textsc{Global Multistage Planar Dominating Set}, as we now require the underlying graph to be planar, while we only required every layer to be planar for \textsc{Global Multistage Planar Dominating Set}.
% However, in this setting we may even restrict our temporal graph to have a planar underlying graph.
% Recall that, for a graph $H$, $D \subseteq E(H)$ is called an \emph{edge dominating set} if
% each edge of $H$ is incident to an edge in $D$.

\newcommand{\gmpeds}{\textsc{Global Multistage Fully Planar Edge Dominating Set}}
\begin{problem}[\gmpeds{}]
        Given a triple $(\Gg,k,\ell)$ where 
        $\mathcal G = \langle G_1, \ldots, G_{\tau} \rangle$ is a temporal graph
		with planar underlying graph $G$,
        the goal is to find sets $S_1, \ldots, S_{\tau}$, each of size at most $k$, such that $S_i \subseteq E(G)$
        is an edge dominating set for $G_i$ ($i\in [\tau]$) and the number of insertions
        $\sum_i \abs{S_{i+1} \setminus S_i}$
        is at most $\ell$.
\end{problem}

\begin{proposition}
	\gmpeds{} is $\wtwo$-hard when parameterized by the solution size $k$ and no relocations are allowed $(\ell = 0)$.
\end{proposition}
\begin{proof}
	We will devise a parameterized reduction from \textsc{Set Cover} here,
	which is well-known to be $\wtwo$-hard~\cite{DF13}.
	A \textsc{Set Cover} instance consists of an integer $n$
	and a family $\Ff \subseteq 2^{[n]}$ of subsets of $[n]$.
	The task is to find a subset $\Ss \subseteq \Ff$ of size at most $\tilde{k}$
	such that $\bigcup \Ss = [n]$.
	We assume \wilog{} that $\bigcup \Ff = [n]$ and $\emptyset \notin \Ff$.
	
	From this, we construct a temporal graph as follows.
	Let $V := \{\star\} \cup [n] \cup \Ff$ be the set of vertices.
	At time step $i \in [n]$, we define the edge set as
	$E(G_i) := \{\{\star, i\}, \{\star, F\} : F \in \Ff \land i \in F\}$.
	For an example see~\cref{fig:cover}.
	Note that the underlying graph $G$ is then a star with center vertex $\star$.
	We claim that $\mathcal{G}$ contains an edge dominating set of size $k$ if and only if
	the \textsc{Set Cover} instance has a solution of size $\tilde{k}$.
	
  \begin{figure}[!t]
    \begin{center}
    	\begin{tikzpicture}
      		\node[vertex, fill=white] (star) at (0, 0) {$\star$};
      		\node[vertex, fill=white] (1) at ($(star) + (-1, 0)$) {$1$};
      		\node[vertex, fill=white] (2) at ($(star) + (1, 0.8)$) {$F_1$};
      		\node[vertex, fill=white] (3) at ($(star) + (1, 0)$) {$F_3$};
      		\node[vertex, fill=white] (4) at ($(star) + (1, -0.8)$) {$F_4$};
			
			\node[] (G1) at ($(star) + (-2, 0)$) {$G_1$};
			
			\draw
				(star) -- (1)
				(star) -- (2)
				(star) -- (3)
				(star) -- (4);
				
      		\node[vertex, fill=white, right= 4 of star] (bstar) {$\star$};
      		\node[vertex, fill=white] (b1) at ($(bstar) + (-1, 0)$) {$2$};
      		\node[vertex, fill=white] (b3) at ($(bstar) + (1, 0)$) {$F_1$};
			
			\node[] (G1) at ($(bstar) + (-2, 0)$) {$G_2$};
			
			\draw
				(bstar) -- (b1)
				(bstar) -- (b3);

      		\node[vertex, fill=white, right= 4 of bstar] (cstar) {$\star$};
      		\node[vertex, fill=white] (c1) at ($(cstar) + (-1, 0)$) {$3$};
      		\node[vertex, fill=white] (c2) at ($(cstar) + (1, 0.8)$) {$F_2$};
      		\node[vertex, fill=white] (c3) at ($(cstar) + (1, 0)$) {$F_3$};
			
			\node[] (G1) at ($(cstar) + (-2, 0)$) {$G_3$};
			
			\draw
				(cstar) -- (c1)
				(cstar) -- (c2)
				(cstar) -- (c3);

      		\node[vertex, fill=white, below= 2 of star] (dstar) {$\star$};
      		\node[vertex, fill=white] (d1) at ($(dstar) + (-1, 0)$) {$4$};
      		\node[vertex, fill=white] (d2) at ($(dstar) + (1, 0.8)$) {$F_2$};
			
			\node[] (G1) at ($(dstar) + (-2, 0)$) {$G_4$};
			
			\draw
				(dstar) -- (d1)
				(dstar) -- (d2);

      		\node[vertex, fill=white, below= 2 of bstar] (estar) {$\star$};
      		\node[vertex, fill=white] (e1) at ($(estar) + (-1, 0)$) {$5$};
      		\node[vertex, fill=white] (e2) at ($(estar) + (1, 0.8)$) {$F_1$};
      		\node[vertex, fill=white] (e3) at ($(estar) + (1, 0)$) {$F_2$};
      		\node[vertex, fill=white] (e4) at ($(estar) + (1, -0.8)$) {$F_4$};
			
			\node[] (G1) at ($(estar) + (-2, 0)$) {$G_5$};
			
			\draw
				(estar) -- (e1)
				(estar) -- (e2)
				(estar) -- (e3)
				(estar) -- (e4);
      \end{tikzpicture}
    \end{center}
    \caption{An example of a temporal graph $\mathcal G = \langle G_1, \ldots, G_5 \rangle$ for the parameterized
    reduction from \textsc{Set Cover} to \textsc{Global Multistage Fully Planar Edge Dominating Set} using the family $\Ff = \{ F_1 = \{1, 2, 5\}, F_2 = \{3, 4, 5\}, F_3 = \{1, 3\}, F_4 = \{1, 5\} \}$. Isolated vertices are not shown.}\label{fig:cover}
  \end{figure}
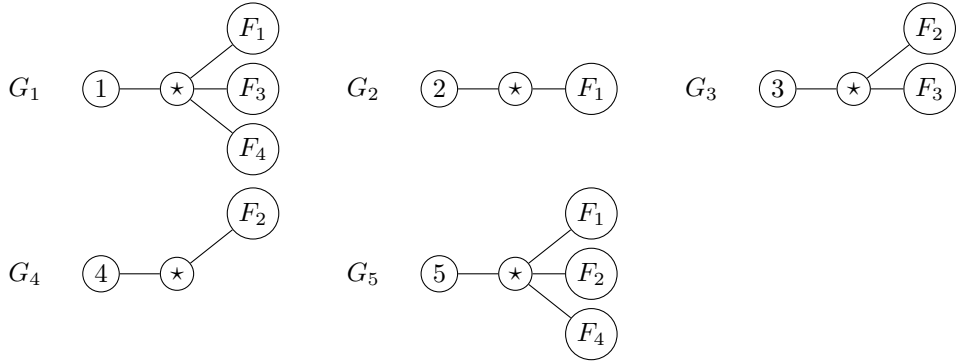

	We may assume that any solution to \gmpeds{} only contains edges of the form $\{\star, F\}$ with $F \in \Ff$
	since any other edge may be exchanged for one of these.
	
	Now the claimed equivalence follows from the fact that $\Ff' \subseteq \Ff$ is a set cover
	if and only if $\{\{\star, F\} \mid F \in \Ff' \}$ is an edge dominating set for each layer of $\Gg$.
\end{proof}

In summary, we have shown that the temporal multistage version of \textsc{Dominating Set} is $\wtwo$-hard for the parameter $k + \ell$ 
even when each layer is a planar graph for \textsc{Dominating Set} or when the underlying graph is a star for \textsc{Edge Dominating Set}.

\section{Parameterized Hardness of Polynomial-Time Solvable Problems Gone Globally Multistage}\label{sect:hardness_of_polytime}

We now present parameterized hardness results for the global multistage versions of several polynomial-time solvable problems.
These hold even if no relocations are allowed ($\ell = 0$) and thus also imply parameterized hardness for the corresponding classical multistage versions of these problems.

\subsection{$s$-$t$-Path}
Contrasting very recent work on \textsc{$s$-$t$-Path} in the 
standard multistage setting~\cite{FNSZ20}, we show that the global 
multistage version of \textsc{$s$-$t$-Path} with no relocations is $\wone$-hard parameterized by the combined parameter solution size~$k$ plus lifetime $\tau$. The problem is formally defined as follows.
\begin{problem}[\textsc{Global Multistage $s$-$t$-Path}]
  Given a temporal graph $\mathcal G = \langle G_1, \ldots, G_{\tau} \rangle$, two vertices $s, t\in V(G)$, and an integer $k$,
  find a set $F\subseteq E(G)$ of at most $k$ edges such that, for every $i\in [\tau]$, there exists an $s$-$t$-path in $G_i' := (V(G_i), E(G_i)\cap F)$.
\end{problem}

\begin{theorem}\label{lem:s-t-path}
  \textsc{Global Multistage $s$-$t$-Path} parameterized by~$k + \tau$ is W[1]-hard, even if~$\ell = 0$.
\end{theorem}

\begin{proof}
	We give a parameterized reduction from the W[1]-hard \textsc{Multicolored Clique} problem~\cite{DF13} %\cite[Theorem~13.25]{CygFKLMPPS15}.
In this problem, a graph $H$ with a partition $(V_1, \dots, V_{k'})$ of $V(H)$ into $k'$ sets is given, and the question is whether $H$ contains a clique containing exactly one vertex from $H_i$ for all $i \in [k']$.

  Let~$(H, k', c)$ be an instance of \textsc{Multicolored Clique}
  where $c: V(H) \rightarrow [k']$ is the coloring determining a $k'$-partition 
   $V_1, \dots, V_{k'}$ of $V(H)$.
  We define a temporal graph $\Gg$ with vertex set $V(\Gg) := V(H) \cup \{s , t\}$ as follows (see also \cref{fig:reachability}).
  For each~$i\in [k']$, the $i$th layer of $\Gg$ contains the edges $\{s, v\}$ and $\{v, t\}$ for each~$v\in V_i$.
  For~$i,j\in [k']$ with $i < j$, we add a layer $L_{i,j}$, which contains the edges $\{s, v\}$ for all~$v\in V_i$, the edges from $E(H[V_i \cup V_j])$, and the edges $\{v, t\}$ for all $v\in V_j$.
  We set $k := 2 k' + \binom{k'}{2}$
  and claim that the resulting \textsc{Global Multistage $s$-$t$-Path} instance $(\Gg, k, 0)$ is solvable if and only if $H$ contains a multicolored $k'$-clique.
  
  $(\Leftarrow)$ Given a clique $\{v_1, \dots , v_{k'}\}$ with $v_i\in V_i$, we construct a solution for $(\Gg, k, 0)$ by taking the edges $\{s, v_i\}$, $\{v_i, t\}$, and $\{v_i, v_j\}$ for all $i, j \in [k']$.
  For $i\in [k']$, the $i$th layer contains the $s$-$t$-path $s$-$v_{i}$-$t$.
  The layer $L_{i,j}$ contains the $s$-$t$-path $s$-$v_i$-$v_j$-$t$.

  $(\Rightarrow)$ Given a solution to $(\Gg, k, 0)$, note that the first $k'$ layers require that for each $i \in [k']$, there is some $v_i\in V_i$ such that $\{s, v_i\}$ and~$\{v_i, t\}$ are contained in the solution.
  The layer $L_{i,j}$ for $i< j$ requires that an edge~$e= \{v_{ij}, w_{ij}\}$ with $v_{ij} \in V_i$ and $w_{ij}\in V_j$ from $E(H[V_i \cup V_j])$ is contained in the solution, and that the edges~$\{s, v_{ij}\}$ and $\{w_{ij}, t\}$ are also contained in the solution.
  Since $k =2k' + \binom{k'}{2}$, no edge but~$\{s, v_i\}$, $\{v_i, t\}$, and $\{v_{ij}, w_{ij}\}$ for $i, j\in [k']$ with $i <j$ is contained in the solution.
  We can conclude that $v_{ij} = v_i$ and $w_{ij} = v_j$ for all $i, j \in [k']$ with $i < j$, implying that $\{v_1, \dots, v_{k'}\}$ is a clique.

  The $\wone$-hardness for \textsc{Global Multistage $s$-$t$-Path} follows directly since $k = 2k' + \binom{k'}{2}$ and $\tau = k' + \binom{k'}{2}$.
\end{proof}

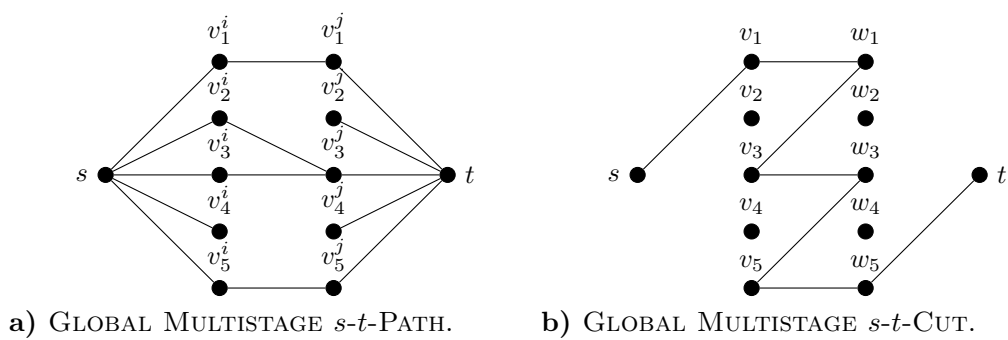
\begin{figure}[t!]
  \begin{subfigure}[t]{0.5\textwidth}
    \centering
	  \begin{tikzpicture}[xscale = 1.5, yscale =0.5]
        \node[vertex, label=180:$s$] (s) at (0, 0) {};
        
        \node[vertex, label=0:$t$] (t) at (3, 0) {};
        
        \node[vertex, label=90:$v_1^i$] (v1) at (1, 3) {};
        \node[vertex, label=90:$v_1^j$] (w1) at (2, 3) {};
        
        \node[vertex, label=90:$v_2^i$] (v2) at (1, 1.5) {};
        \node[vertex, label=90:$v_2^j$] (w2) at (2, 1.5) {};
        
        \node[vertex, label=90:$v_3^i$] (v3) at (1, 0) {};
        \node[vertex, label=90:$v_3^j$] (w3) at (2, 0) {};
        
        \node[vertex, label=90:$v_4^i$] (v4) at (1, -1.5) {};
        \node[vertex, label=90:$v_4^j$] (w4) at (2, -1.5) {};
        
        \node[vertex, label=90:$v_5^i$] (v5) at (1, -3) {};
        \node[vertex, label=90:$v_5^j$] (w5) at (2, -3) {};
        
        \draw (s) -- (v1) -- (w1) -- (t) -- (w3) -- (v3) -- (s) -- (v5) -- (w5) -- (t);
        \draw (s) -- (v2)-- (w3);
        \draw (s) -- (v4);
        \draw (w2) -- (t);
        \draw (w4) -- (t);
      \end{tikzpicture}
    \caption{\textsc{Global Multistage $s$-$t$-Path}.}\label{fig:reachability}
  \end{subfigure}\hfill
  \begin{subfigure}[t]{0.5\textwidth}
    \centering
      \begin{tikzpicture}[xscale = 1.5, yscale =0.5]
        \node[vertex, label=180:$s$] (s) at (0, 0) {};
        
        \node[vertex, label=0:$t$] (t) at (3, 0) {};
        
        \node[vertex, label=90:$v_1$] (v1) at (1, 3) {};
        \node[vertex, label=90:$w_1$] (w1) at (2, 3) {};
        
        \node[vertex, label=90:$v_2$] (v2) at (1, 1.5) {};
        \node[vertex, label=90:$w_2$] (w2) at (2, 1.5) {};
        
        \node[vertex, label=90:$v_3$] (v3) at (1, 0) {};
        \node[vertex, label=90:$w_3$] (w3) at (2, 0) {};
        
        \node[vertex, label=90:$v_4$] (v4) at (1, -1.5) {};
        \node[vertex, label=90:$w_4$] (w4) at (2, -1.5) {};
        
        \node[vertex, label=90:$v_5$] (v5) at (1, -3) {};
        \node[vertex, label=90:$w_5$] (w5) at (2, -3) {};
        
        \draw (s) -- (v1) -- (w1)  -- (v3) -- (w3) -- (v5) -- (w5) -- (t);
      \end{tikzpicture}
    \caption{\textsc{Global Multistage $s$-$t$-Cut}.}
    \label{fig:cut}
  \end{subfigure}
  \caption{
  {\bfseries a)} An example for the layer $L_{ij}$ for the reduction to \textsc{Global Multistage $s$-$t$-path}, where $V_i = \{v_1^i, \dots, v_5^i\}$, $V_j =  \{v_1^j, \dots, v_5^j\}$, and $E(H[V_i \cup V_j]) = \{\{v_1^i, v_1^j\},\{v_2^i, v_3^j\},\{v_3^i, v_3^j\},\{v_5^i, v_5^j\}\}$.
  {\bfseries b)}
  An example for the $j$-th layer of $\mathcal{G}$ for the reduction to \textsc{Global Multistage $s$-$t$-Cut} for a subset~$S_j = \{u_1, u_3, u_5\} \subseteq U$ with $U = \{u_1, u_2, u_3, u_4, u_5\}$ .}\label{fig:example}
\end{figure}

\subsection{$s$-$t$-Cut}
We now consider the problem of separating the vertices $s$~and~$t$ in each layer. We show that the global multistage version of \textsc{$s$-$t$-Cut} with no relocations problem is $\wtwo$-hard parameterized by solution size~$k$.

\begin{problem}[\textsc{Global Multistage $s$-$t$-Cut}]

  Given a temporal graph $\mathcal G = \langle G_1, \ldots, G_{\tau} \rangle$, two vertices $s$ and $t$, and an integer $k$, decide whether 
there exists a set $F\subseteq E(\mathcal{G})$ of at most~$k$ edges such that, for every $i\in [\tau]$, there exists no $s$-$t$-path in $G_i' := (V(G_i), E(G_i)\setminus  F)$.  
The set $F$ is called a \emph{temporal $s$-$t$-cut}.
\end{problem}

We show that this problem is $\wtwo$-hard by a parameterized reduction similar to the one for \textsc{Global Multistage $s$-$t$-Path}.

  \begin{theorem}\label{lem:forward-cut}
      \textsc{Global Multistage $s$-$t$-Cut} parameterized by~$k$ is W[2]-hard, even if~$\ell = 0$.
  \end{theorem}
  \begin{proof}
	  We give a parameterized reduction from the $\wtwo$-hard \textsc{Hitting Set} problem~\cite{DF13}.  % \cite[Theorem~13.21]{CygFKLMPPS15}.
  Here we are given a set~$U$, a set familiy~$\Ss \subseteq 2^U$, and an integer~$k'$ and need to find a subset of at most $k'$~elements of~$U$ which intersects every member of $\Ss$.
  
  Given an instance $\mathcal{I} = (U, \mathcal{S}, k')$ of \textsc{Hitting Set}, we design an instance of \textsc{Global Multistage $s$-$t$-Cut} as follows (see also \cref{fig:cut}).
  The graph $\mathcal{G}$ contains two vertices~$s$ and~$t$.
  Furthermore, for each $u_i\in U$, there are two vertices $v_i$ and $w_i$.  
  For each set~$S_j = \{u_{i_1}, u_{i_2}, \dots, u_{i_r}\}\in \mathcal{S}$, the $j$-th layer contains the $s$-$t$-path $s$-$v_{i_1}$-$w_{i_1}$-$v_{i_2}$-$w_{i_2}$-$v_{i_3}$-$\dots$-$v_{i_r}$-$w_{i_r}$-$t$.
  We set~$k:= k'$.
  We claim that the resulting \textsc{Global Multistage $s$-$t$-Cut} instance $(\Gg, k, 0)$ is equivalent to the given \textsc{Hitting Set} instance $(U, \Ss, k')$.

    $(\Leftarrow)$ Let $X = \{u_{i_1}, u_{i_2}, \dots, u_{i_k}\}$ be a hitting set.
    We define $F:= \{\{v_{i_j}, w_{i_j}\} : j\in [k]\}$, and claim that $F$ separates $s$ and $t$ in each layer.
    
    Consider the $j$-th layer of $\mathcal{G}$.
    This layer contains a unique $s$-$t$-path, which contains the edge $\{v_i, w_i\}$ for each $u_i\in S_j$.
    As $X$ is a hitting set, there is an edge of this path contained in $F$ and, thus, $F$ is a temporal $s$-$t$-cut.

	$(\Rightarrow)$ Let $F$ separate $s$ and $t$ in each layer and let $|F| \leq k$.
    If $F$ contains an edge of the form $\{s, v_i\}$, $\{v_i, w_j\}$ for $j\neq i$, or $\{w_i, t\}$, then we can replace this edge by $\{v_i, w_i\}$, as any $s$-$t$-path containing such an edge contains also the edge $\{v_i, w_i\}$.
    Thus, we assume that $F= \{\{v_i, w_i\} : i\in I\}$ for some $I\subseteq [n]$ with $|I| = k$.
    
    Let $X := \{u_i : i\in I\}$.
    Clearly, $|X|=|I| = k$, so it remains to show that $X$ is a hitting set.
    Consider a set $S_j$.
    Since $F$ is an $s$-$t$-cut in $G_j$, it contains an edge $e = \{v_i, w_i\}$ of the unique $s$-$t$-path in this layer.
    Thus, we have $u_i\in X \cap S_j$ and, therefore, $X$ hits $S_j$.
    
    The $\wtwo$-hardness for \textsc{Global Multistage $s$-$t$-Cut} follows directly since $k = k'$.
  \end{proof}

\subsection{Matching}
Finally, we show that the global multistage version of \textsc{Matching} with no relocations is $\wone$-hard parameterized by solution size~$k$.
This stands in contrast to recent fixed-parameter tractability results 
for a (non-multistage) temporal version of \textsc{Matching} based 
on time windows~\cite{BBR20,MMNZZ20}.
The problem is formally defined as follows.
\begin{problem}[\textsc{Global Multistage Matching}]
  Given a temporal graph $\mathcal G = \langle G_1, \ldots, G_{\tau} \rangle$, and an integer $k$,
  decide whether there exists a set $F\subseteq E(\mathcal{G})$ of at least~$k$ edges such that, for every~$i\in [\tau]$,
  the elements of $E(G_i)\cap F$ are pairwise disjoint.
  The set $F$ is called a \emph{temporal matching}.
\end{problem}

  \begin{figure}[!t]
    \begin{center}
      \begin{tikzpicture}
        \node[vertex, label=180:$c$] (c) at (0, 0) {};

 \foreach \x in {0,60,...,300} {
    }

        \node[vertex, label=0:$v_j$] (vj) at (0:1cm) {};
        \node[vertex, label=72:$v$] (v) at (72:1cm) {};
        
        \node[vertex, label=144:$v'$] (v') at (144:1cm) {};
        \node[vertex, label=216:$w_j$] (wj) at (216:1cm) {};
        
        \node[vertex, label=288:$v''$] (v') at (288:1cm) {};

        \draw (wj) -- (c) -- (vj);
      \end{tikzpicture}

    \end{center}
    \caption{An example for the $i$-th layer of $\mathcal{G}$ for the \textsc{Global Multistage Matching} reduction from \textsc{Independent Set} for an edge~$e_i = \{v_i, w_i\}$ and $V(H) = \{v_i, w_i, v, v', v''\}$.}
    \label{fig:matching}
  \end{figure}
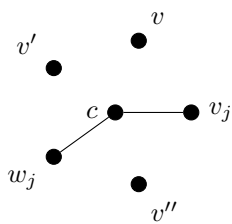

\begin{theorem}\label{lem:forward-matching}
      \textsc{Global Multistage Matching} parameterized by~$k$ is W[1]-hard, even if~$\ell = 0$.
\end{theorem}
\begin{proof}
	We give a parameterized reduction from the \textsc{Independent Set} problem.
	Given a graph $H$ and an integer $k'$, the problem \textsc{Independent Set} asks whether $H$ contains a set of pairwise non-adjacent vertices of size at least $k'$.
	It is well-known to be $\wone$-hard parameterized by solution size \cite{DF13}.
	
	Given an instance $(H, k')$ of \textsc{Independent Set}, we construct an instance of \textsc{Global Multistage Matching} as follows.
	Let $E(H) =: \{e_1, e_2, \dots, e_m\}$ with $e_i =: \{v_i, w_i\}$.
	The vertices of $\mathcal{G}$ are $V(\mathcal{G}) := V(H) \cup \{c\}$, where $c$ is not contained in $V(H)$.
	The graph~$\mathcal{G}$ has $m$~layers.
	The $i$-th layer of $\mathcal{G}$ contains the edges $\{v_i, c\}$ and $\{w_i, c\}$, see \cref{fig:matching} for an example.
	Finally, we set $k:= k'$.

  We claim that the resulting \textsc{Global Multistage Matching} instance $(\Gg, k, 0)$ is equivalent to the \textsc{Independent Set} instance $(H, k')$.
  
  $(\Leftarrow)$ Let $X$ be an independent set in $H$.
  Define $M:=\{\{v, c\}: v\in X\}$.
  Clearly, $|M| = k$, so it remains to show that $M\cap E(G_i)$ is a matching for all $i\in [\tau]$.
  Every layer containing an edge $\{v_i, c\}$ contains only one other edge, namely $\{w_i, c\}$.
  If $v_i \in X$ ($w_i\in X$ is symmetric), then $w_i\notin X$, as $X$ is an independet set.
  Thus, $F\cap E(G_i) = \{\{v_i, c\}\}$ is a matching.
  If neither $v_i\in X$ nor $w_i\in X$, then $F\cap E(G_i) = \emptyset$ is a matching.

  $(\Rightarrow)$ Let $F\subseteq E(\mathcal{G})$ be a temporal matching of size $k$.
  Define $X:= \{v\in V(H) : \{v, c\} \in F\}$.
  Clearly, $|X| = k$.  
  It remains to show that $X$ is an independent set.
  So assume that $e_i\in E(H[X])$.
  Then~$\{v_i, c\}\in F \cap E(G_i)$ and $\{w_i, c\}\in F \cap E(G_i)$, contradicting the assumption that $F$ is a temporal matching.
  
  Now the $\wone$-hardness for \textsc{Global Multistage Matching} follows directly since $k' = k$.
\end{proof}

\section{Conclusion}\label{sect:conclusion}
We described a general approach to derive fixed-parameter 
tractability results (including polynomial-size problem kernels) for
global multistage versions of classical NP-hard problems.
A particular technical feature herein is showing how to derive \fpt-algorithms for global multistage problems from \fpt-enumeration algorithms for their static counterpart and
``temporal kernels'' from known static ones (more specifically, 
known full kernels~\cite{DAMASCHKE2006337}---a method that to the 
best of our knowledge has not yet been used in the design of algorithms for 
temporal (graph) problems).
Our results are complemented by several parameterized hardness results, 
indicating that the parameter solution size alone does not suffice 
to obtain fixed-parameter tractability results.
Furthermore, we could show that for some problems which do not fit our framework,
the global multistage versions are hard in a parameterized sense,
even though they are polynomial-time solvable in the static case.
We remark that
all our hardness results hold for the variant of the global multistage scenario where the first solution~$S_1$ is given as part of the
input.

As to challenges for future work, we also refer to the open questions exhibited in Table~\ref{tab-results}.
In particular, the parameterized complexity with respect to the combined parameter $\tau +k$
is still open for some of the problems studied. 
{F}rom a more general perspective, exploring connections to 
dynamic problems~\cite{AEFRS15,HN13,KST18,LMNN18} and 
reoptimization~\cite{BBRR18} may be fruitful as well.
Finally, studying (further) temporal problems in the global multistage setting 
may be of general interest in the area of complex network analysis.

\bibliography{main}
\appendix

\end{document}